\documentclass[12pt]{article}
\usepackage[T1]{fontenc}
\usepackage[utf8]{inputenc}
\usepackage{lmodern}
\usepackage{microtype}
\usepackage{amsmath,amssymb,amsthm,mathtools}
\usepackage{cite}
\usepackage{authblk}
\usepackage{geometry}
\usepackage{graphicx}
\usepackage[title]{appendix}
\usepackage{xcolor}
\usepackage{hyperref}
\hypersetup{
    colorlinks=true,
    linkcolor=black,
    citecolor=black,
    urlcolor=black
}
\usepackage{orcidlink}
\theoremstyle{plain}
\newtheorem{theorem}{Theorem}
\numberwithin{theorem}{section}
\newtheorem{proposition}[theorem]{Proposition}
\newtheorem{lemma}[theorem]{Lemma}
\newtheorem{corollary}[theorem]{Corollary}
\theoremstyle{definition}
\newtheorem{definition}[theorem]{Definition}
\newtheorem{remark}[theorem]{Remark}
\newtheorem{conjecture}[theorem]{Conjecture}
\DeclareMathOperator{\Dom}{Dom}

\begin{document}

\title{A general proof of integer R\'{e}nyi QNEC}

\author[1]{Tanay Kibe\,\orcidlink{0000-0002-0571-1712} \thanks{tanay.kibe@ib.edu.ar}}
\affil[1]{Instituto Balseiro, Centro At\'omico Bariloche,
S.C. de Bariloche, 8400 R\'io Negro, Argentina}

\author[2]{Pratik Roy\,\orcidlink{0000-0003-4327-3132}%
\thanks{%
\href{mailto:roy.pratik92@gmail.com}{roy.pratik92@gmail.com}}}

\affil[2]{Institute of Mathematics, University of Warsaw,
ul. Banacha 2, 02-097 Warsaw, Poland}

\date{}

\maketitle

\begin{abstract}
We prove the R\'{e}nyi quantum null energy condition for
every integer R\'{e}nyi order \(n\geq2\). In local Poincar\'e-invariant
quantum field theory, this states that the sandwiched R\'{e}nyi
divergence (SRD) of an excited state relative to the vacuum is convex
along a future-directed null deformation. We assume that the
deformation is realized by the canonical translations of a half-sided
modular inclusion and that the divergence is finite for the initial, undeformed algebra. 
The proof establishes convexity of the SRD along the induced positive contraction semigroup on Haagerup \(L^n(M)\), for every non-zero positive initial element.
We also establish a R\'enyi ant formula at every even integer order and,
under an additional regularity assumption, at odd integer orders. The
R\'enyi ant formula identifies the negative first null derivative of SRD with a
multiple of the infimum of the canonical null-translation energy over
vector purifications of an escort state constructed from the excited
state and the vacuum.
\end{abstract}

% {\small\medskip
% \noindent\textbf{Keywords:}
% Quantum null energy condition;
% sandwiched R\'enyi divergence;
% half-sided modular inclusions;
% noncommutative $L^p$ spaces;
% modular theory;
% averaged null energy.

% \smallskip
% \noindent\textbf{Mathematics Subject Classification (2020):}
% Primary 81T05;
% Secondary 81P17, 46L10, 46L52.}

\newpage

\tableofcontents

\section{Introduction}
\label{sec:introduction}

The quantum null energy condition (QNEC) is a local relation between
energy and quantum information in quantum field theory (QFT). It was first obtained as the nongravitational limit of the quantum focusing conjecture \cite{Bousso:2015mna}, and
has since been established in free and holographic theories and, under broad
assumptions, in interacting QFT
\cite{Bousso:2015wca,Malik:2019dpg,Koeller:2015qmn,
Balakrishnan:2017bjg}.
Beyond serving as a consistency condition on
relativistic QFT, it constrains the flow of quantum information along null
directions.  These constraints have found applications to entanglement
growth and entropy production in nonequilibrium processes
\cite{Mezei:2019sla,Kibe:2021qjy,Kibe:2024icu}, and to
renormalization-group irreversibility in the presence of defects
\cite{Casini:2023kyj}.

The information-theoretic content of the QNEC is particularly transparent
for null cuts, i.e., cross-sections of a null plane.  Moving a cut towards the future removes
observables from the accessible region and hence produces a decreasing
family of local von Neumann algebras \(M_c\) labelled by a real deformation parameter $c$.  If \(\psi\) is an excited
state and \(\omega\) is the vacuum state, the data-processing inequality implies that the relative entropy \(D(\psi\Vert\omega;M_c)\) cannot increase with \(c\).
The relative-entropy formulation of the QNEC is the further
assertion that this function is convex along the null-deformation parameter. Schematically,
\begin{equation}
    \frac{d^2}{dc^2}D(\psi\Vert\omega;M_c)\geq0.
    \label{eq:intro-relative-entropy-qnec}
\end{equation} 
Its connection with energy comes from
the vacuum modular Hamiltonian, which for null cuts is expressed in terms of
the null stress tensor \cite{Casini:2017roe}.  

Local algebras in QFT are generally of Type~III, so the reduced density
matrices familiar from finite quantum systems do not exist.  Araki relative entropy is instead
defined directly for normal states on a von Neumann algebra, with no need for an ultraviolet regulator
\cite{Araki:1976zv,Araki:1977zsq}.  For null cuts admitting geometric
vacuum modular flow, the nesting of the observable algebras and the
rescaling action of this flow form a half-sided modular inclusion (HSMI) \cite{Borchers:1991xk,Wiesbrock:1992mg,Borchers:1995zg,
Borchers1996,Araki:2005we}.  Within this framework, Wall's
information-theoretic argument \cite{Wall:2017blw} has been made precise in
model-independent proofs of the QNEC
\cite{Ceyhan:2018zfg,Hollands:2025glm}.

We ask whether null convexity is special to relative entropy or also holds
for other quantum distinguishability measures.  The sandwiched R\'{e}nyi
divergences (SRD) \(D_\alpha(\psi\Vert\omega)\) are a particularly relevant family, which form a one-parameter refinement of relative
entropy and obey data processing throughout the upper-order range \(\alpha>1\)
\cite{Beigi:2013dpi,Frank:2013mre}.  Moreover, if
\(D_{\alpha_0}(\psi\Vert\omega)\) is finite for some \(\alpha_0>1\), then
\(D_\alpha(\psi\Vert\omega)\) converges to relative entropy as
\(\alpha\downarrow1\)
\cite{Muller-Lennert:2013liu,Wilde:2013bdg,Jencova:2016tqz}.  Their intrinsic
extension to normal positive functionals on arbitrary von Neumann algebras
is well established
\cite{Berta:2016vnw,Jencova:2016tqz,Jencova:2017txf}.  Data processing of $D_\alpha$ gives
monotonicity under restriction to subregions. Lashkari formulated the question of its convexity under null deformations as the R\'{e}nyi quantum null energy
condition (RQNEC) \cite{Lashkari:2018nsl, Moosa:2020jwt}.  Null-quantization methods have
established the RQNEC for \(\alpha>1\) in free and superrenormalizable
bosonic and fermionic theories in dimensions \(d>2\)
\cite{Moosa:2020jwt,Roy:2022yzm}.  Counterexamples below \(\alpha=1\) are
known \cite{Moosa:2020jwt}.  These results leave open whether upper-order
RQNEC follows from the general modular structure of null deformations,
without a free-field realization.

A recent quasi-local formulation expresses the RQNEC as an
energy--information bound, with the energy evaluated in an escort-averaged state
constructed from the excited and reference states
\cite{Kibe:2026idb}.  This formulation further motivates studying the inequality
beyond the theories accessible to null quantization, as well as its
consequences for energy and entropy transport.

In this paper we prove the RQNEC for every integer R\'{e}nyi order
\(n\geq2\).  Let \(M_c\) be the decreasing family of algebras generated by
the canonical null translations of a half-sided modular inclusion, let
\(\omega\) be its distinguished faithful normal state, and let \(\psi\) be a
non-zero normal positive functional.  We prove as Corollary \ref{cor:integer-renyi-qnec-fixed} that, whenever
\(D_n(\psi\Vert\omega;M_0)<\infty\), the function
\begin{equation}
    S_n(c)\coloneqq D_n(\psi\Vert\omega;M_c),
    \qquad c\geq0,
    \label{eq:intro-integer-rqnec}
\end{equation}
is finite, continuous, non-increasing and convex on \([0,\infty)\), and real analytic on $(0,\infty)$.
Equivalently, for every $c>0$, \(S_n''(c)\geq0\).  When this
canonical family coincides with the geometric null-cut family in a
Poincar\'e-invariant QFT, this is the RQNEC along the chosen non-negative
null-deformation profile. This result requires no
additional differentiability or null-translation generator-domain
condition.  At integer order, the RQNEC thus follows from the
modular-translation structure without any free-field realization or
explicit stress-tensor representation.  The extension to arbitrary real
\(\alpha>1\) remains open.

Our proof begins by demonstrating that the relevant HSMI null deformations 
are implemented by a positive, strongly continuous
contraction semigroup on Haagerup's noncommutative \(L^p\)-space.
Its standard-form realization at $p=2$,
together with interpolation, shows that this semigroup is holomorphic
for every $1<p<\infty$. Every orbit of the semigroup is 
thus norm analytic on $(0,\infty)$,
without any generator-domain assumption on the initial \(L^p\) element.
At integer order, modular covariance yields a cyclic Ward identity that
is used to express the second derivative of the SRD in terms of
\(L^2\)-norms. Cauchy--Schwarz then gives convexity at every
positive cut, and strong continuity extends the result to the initial
cut.

We also prove an integer-order R\'{e}nyi version of the
ant formula for relative entropy
\cite{Wall:2017blw,Ceyhan:2018zfg,Hollands:2025glm}.
For finite initial $D_n$ and a positive cut $c$, the
formula identifies $-S_n'(c)$ with
$\kappa_n=2n\sin(\pi/n)/(n-1)$ times the infimum of the
canonical translation energy over vector purifications
of an escort state constructed from the two functionals $\psi,\omega$ entering the SRD. Finiteness of SRD at the initial
cut suffices for every even $n\geq2$.
For odd $n$, we require an additional domain inclusion assumption, stated precisely in Section~\ref{sec:ant-formula}. 
When the canonical translation generator is the weighted averaged null-energy
operator, this gives a variational energy interpretation
of the first null derivative of the SRD. We conjecture
that the same formula holds for every real $\alpha>1$
and verify it for a coherent-state family.

The rest of this paper is organized as follows.
Section~\ref{sec:null-deformations-main-results} gives
the geometric and algebraic formulations of the main results.
Section~\ref{sec:Ln-dynamics} constructs the null translation semigroup, 
establishes its standard-form realization and interior regularity,
and identifies its orbits with translated sandwiched densities.
Section~\ref{sec:integer-renyi-qnec} proves the cyclic Ward
identity and log-convexity, characterizes saturation at every
integer order, and establishes strict decrease and strict convexity for
nonstationary orbits. It also gives an independent spectral
proof at order two.
Section~\ref{sec:ant-formula} establishes the integer-order
R\'enyi ant formula and presents the general-order conjecture
and coherent-state check.
Section~\ref{sec:discussion} discusses extensions.
Appendices~\ref{appsec:Haagerup-analytic-facts}
and~\ref{appsec:Ward-core-construction} contain the analytic
facts and test-core construction used in the proof.

\section{Null deformations and the main results}
\label{sec:null-deformations-main-results}

In this section, we formulate the null-deformation problem in terms of a half-sided modular
inclusion.  After recalling the null-cut geometry, we identify the moving
algebras with a fixed algebra \(M\).  R\'{e}nyi QNEC then becomes a
log-convexity problem for the HSMI semigroup implemented on Haagerup's noncommutative space \(L^p(M)\). The stress-tensor formulas in Section~\ref{subsec:null-cuts}
provide the physical setting.  The proof itself uses only the abstract HSMI relations stated in
Section~\ref{subsec:HSMI-null-translations} and does not require these formulas. 

%%%%%%%%%%%%%%%%%%%%%%%%%%%%%%%%%%%%%%%%%%%%%%%%%%%%%%%%%%%%%%%%%%%%%%%%%%%%%%
\subsection{Null cuts and R\'{e}nyi QNEC}
\label{subsec:null-cuts}

Consider a Poincar\'{e}-invariant local QFT on \(d\)-dimensional
Minkowski spacetime, and let \(\mathcal A(\mathcal O)\)
denote the von Neumann algebra of observables localized in a spacetime
region \(\mathcal O\) and acting on a Hilbert space $\mathcal H$.  We use null coordinates \((x^+,x^-,\vec y)\), with metric
\begin{equation}
    ds^2=-dx^+\,dx^-+d\vec y^{\,2},
    \label{eq:null-coordinate-metric}
\end{equation}
where \(\vec y\in\mathbb R^{d-2}\) labels the transverse null generators.
The right Rindler wedge \(R=\{x^+>0,\ x^-<0\}\) has future horizon on the
null plane \(x^-=0\). We write \(M_R=\mathcal A(R)\).

A continuous function \(C:\mathbb R^{d-2}\to\mathbb R\) specifies a cut of
this null plane.  Its future portion and associated algebra are
\begin{equation}
    N_C\coloneqq\bigl\{x^-=0,\ x^+>C(\vec y)\bigr\},
    \qquad
    M_C\coloneqq\mathcal A(N_C''),
\end{equation}
where \(N_C''\) denotes the causal completion.  The zero cut
\(C\equiv0\) gives the Rindler-wedge algebra \(M_R\).  Isotony turns the
pointwise order on cuts into reverse inclusion of algebras:
\begin{equation}
    C_2(\vec y)\geq C_1(\vec y)
    \quad\Longrightarrow\quad
    M_{C_2}\subseteq M_{C_1}.
    \label{eq:null-cut-isotony}
\end{equation}
Thus a future-directed displacement of the cut removes observables from the
accessible algebra.  See Figure~\ref{fig:nullcuts} for an illustration.

\begin{figure}[t]
    \centering
    \includegraphics[width=0.52\linewidth]{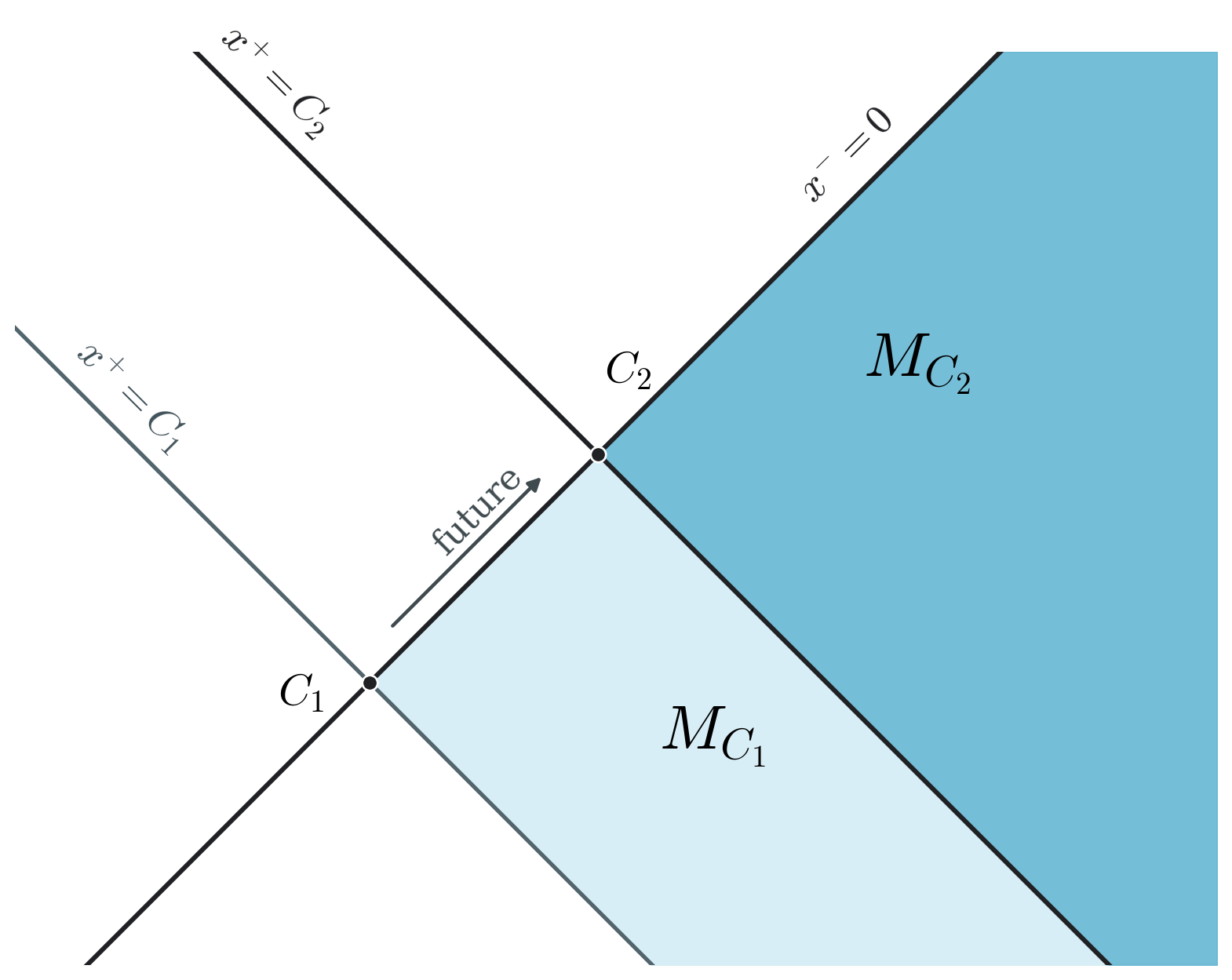}
    \caption{A null deformation of the Rindler wedge.  The transverse
    directions are suppressed. The relevant algebras consist of bounded operators supported in the blue regions. Moving the cut towards the future produces
    a smaller algebra (dark blue): if \(C_2\geq C_1\), then
    \(M_{C_2}\subset M_{C_1}\).}
    \label{fig:nullcuts}
\end{figure}

Let \(\Omega\in\mathcal H\) be the vacuum vector and
\(\omega(a)=\langle\Omega,a\Omega\rangle\) its state.  We assume that
\(\Omega\) is cyclic and separating for the null-cut
algebras under consideration.  The corresponding modular operators are
therefore well defined.  In vacuum representations satisfying the
Bisognano--Wichmann property, modular flow for the Rindler wedge is the
Lorentz-boost flow preserving its edge \cite{Bisognano:1975ih}.  When the geometric null-cut
modular flow described below is available, it acts along the null generators
and rescales their affine distance from the cut
\cite{Casini:2017roe}.

Let
\(K_C\coloneqq-\log\Delta_{\Omega;M_C}\) denote the full vacuum modular
Hamiltonian, so that
\(\Delta_{\Omega;M_C}^{is}=e^{-isK_C}\).  In the class of theories and for
the admissible cuts covered by the null-plane modular-Hamiltonian
construction of \cite{Casini:2017roe}, one has
\begin{equation}
\begin{aligned}
    K_C
    ={}&2\pi
    \int_{\mathbb R^{d-2}}d^{d-2}\vec y
    \int_{-\infty}^{\infty}dx^+\,
    \bigl(x^+-C(\vec y)\bigr)
    T_{++}(x^+,0,\vec y),
\end{aligned}
    \label{eq:null-cut-full-modular-Hamiltonian}
\end{equation}
where $T_{++}(x^+,x^-,\vec y)$ is the null-null component of the energy-momentum tensor.

Choose two cuts \(A\) and \(B\) with \(B(\vec y)\geq A(\vec y)\), and
move along the fixed non-negative deformation profile \(B-A\).  We consider the one parameter family of smooth cuts
\begin{equation}
    A_c(\vec y)
    \coloneqq A(\vec y)+c\bigl(B(\vec y)-A(\vec y)\bigr),
    \qquad 
    c\geq0,
\end{equation}
and define the family of algebras
\begin{equation} \label{eq:null-cut-family}
     M_{A_0}=M_A,
    \qquad M_{A_1}=M_B,
    \qquad 
    M_{A_{c_2}}\subset M_{A_{c_1}},
    \qquad c_2> c_1\geq0.
\end{equation}
% \begin{equation}
% \begin{aligned}
%     A_c(\vec y)
%     &\coloneqq A(\vec y)+c\bigl(B(\vec y)-A(\vec y)\bigr),
%     &&c\geq0,\\
%     M_{A_0}&=M_A,
%     \qquad M_{A_1}=M_B,\\
%     M_{A_{c_2}}&\subset M_{A_{c_1}},
%     &&c_2> c_1\geq0.
% \end{aligned}
%     \label{eq:null-cut-family}
% \end{equation}
The final inclusion follows directly from isotony,
Eq.~\eqref{eq:null-cut-isotony}.

When \((M_B\subset M_A,\Omega)\) has the half-sided modular-inclusion
structure introduced below, its canonical translation generator, defined
intrinsically by modular theory in \eqref{eq:hsmi-positive-generator}, is
represented in this stress-tensor realization, whenever the smearing is
well defined, by
\begin{equation}
\begin{aligned}
    P
    ={}&
    \int_{\mathbb R^{d-2}}d^{d-2}\vec y\,
    \bigl(B(\vec y)-A(\vec y)\bigr)
    \int_{-\infty}^{\infty}dx^+\,
    T_{++}(x^+,0,\vec y).
\end{aligned}
    \label{eq:weighted-ane-generator}
\end{equation}
Thus \(P\) is the averaged null energy, with transverse smearing profile
\(B-A\).  Formally, the stress-tensor expressions give
\(2\pi P=K_A-K_B\).  The abstract modular-inclusion theory in the next
subsection defines this generator without taking an ill-defined difference
of unbounded operators, and recovers both its positivity and its action on
the nested algebras without reference to the stress-tensor formulas above.

Let \(\psi\in(M_A)_*^+\) be a non-zero normal positive functional
describing an excited state.  We use the same symbols \(\psi\) and
\(\omega\) for their restrictions to the algebras \(M_{A_c}\).  For
\(\alpha>1\), set
\begin{equation}
    S_\alpha^{\mathrm{geom}}(c)
    \coloneqq
    D_\alpha(\psi\Vert\omega;M_{A_c}),
    \label{eq:null-cut-Salpha}
\end{equation}
where \(D_\alpha\) is the sandwiched R\'{e}nyi divergence (SRD), defined in
Section~\ref{subsec:fixed-algebra-SRD}.  The one-parameter R\'{e}nyi QNEC
proposed in \cite{Lashkari:2018nsl} asserts that
\(S_\alpha^{\mathrm{geom}}\) is convex
along every future-directed null-deformation ray.  In the present setting,
this means
\begin{equation}
\begin{aligned}
    S_\alpha^{\mathrm{geom}}
    \bigl((1-\theta)c_0+\theta c_1\bigr)
    \leq
    (1-\theta)S_\alpha^{\mathrm{geom}}(c_0)
    +\theta S_\alpha^{\mathrm{geom}}(c_1),
\end{aligned}
    \label{eq:rqnec-global-convexity}
\end{equation}
for \(c_0,c_1\geq0\) and \(\theta\in[0,1]\).  At every \(c>0\) at which
\(S_\alpha^{\mathrm{geom}}\) is twice differentiable, convexity implies
\begin{equation}
    \frac{d^2}{dc^2}
    D_\alpha(\psi\Vert\omega;M_{A_c})
    \geq0.
    \label{eq:rqnec-local-form}
\end{equation}
For a normalized state \(\psi\), suppose that
\(D_{\alpha_0}(\psi\Vert\omega;M_A)<\infty\) for some
\(\alpha_0>1\).  Data processing gives the same finiteness on every later
cut, and the upper-order limit theorem
\cite[Proposition~3.8(ii)]{Jencova:2016tqz} yields
\begin{equation}
    \lim_{\alpha\downarrow1}
    D_\alpha(\psi\Vert\omega;M_{A_c})
    =D(\psi\Vert\omega;M_{A_c}),
    \qquad c\geq0.
    \label{eq:SRD-relative-entropy-limit}
\end{equation}
Consequently, the \(\alpha=1\) endpoint of the global convexity statement
\eqref{eq:rqnec-global-convexity} is the relative-entropy convexity
formulation of the ordinary QNEC
\cite{Ceyhan:2018zfg,Hollands:2025glm}.

When the geometric cut algebras coincide
with the canonical translation family of an HSMI, our result establishes
this convexity for every integer \(\alpha=n\geq2\).  The rest of this
section makes the statement precise.

%%%%%%%%%%%%%%%%%%%%%%%%%%%%%%%%%%%%%%%%%%%%%%%%%%%%%%%%%%%%%%%%%%%%%%%%%%%%%%
\subsection{Half-sided modular inclusions and null translations}
\label{subsec:HSMI-null-translations}

The null-cut geometry has two essential features: the algebras form an inclusion
as the cut moves towards the future, and vacuum modular flow rescales the
affine distance between cuts.  These
properties are expressed algebraically by a half-sided modular inclusion.  We use the
modular-theory conventions of \cite{Hiai:2021qfv}.

\begin{definition}[Half-sided modular inclusion]
\label{def:HSMI}
Let \(M_B\subset M_A\) be von Neumann algebras, and
let \(\Omega\in\mathcal H\) be a unit vector that is cyclic and separating
for both algebras.  We write
\(\omega(a)=\langle\Omega,a\Omega\rangle\) for the resulting faithful
normal state on \(M_A\), and use the same symbol for its restriction to
\(M_B\).  With our sign convention, the triple
\((M_B\subset M_A,\Omega)\) is a half-sided modular inclusion (HSMI) if
\begin{equation}
    \Delta_{\Omega;M_A}^{-is}
    M_B
    \Delta_{\Omega;M_A}^{is}
    \subset M_B,
    \qquad s>0.
    \label{eq:hsmi-condition}
\end{equation}
\end{definition}

An HSMI carries a canonical one-parameter unitary
group, the Borchers--Wiesbrock translations.  We denote it by \(U_b\) and,
in our null-cut setting, call its elements null translations.  The following
theorem records the properties used below.  For proofs, see
\cite{Borchers:1991xk,Wiesbrock:1992mg,Borchers:1995zg,
Borchers1996,Araki:2005we}.

\begin{theorem}[Borchers--Wiesbrock translation structure]
\label{thm:modulartranslations}
Let \((M_B\subset M_A,\Omega)\) be a half-sided modular inclusion.  Then
\begin{equation}
    P
    \coloneqq
    \frac{1}{2\pi}
    \overline{
      \bigl(
        \log\Delta_{\Omega;M_B}
        -
        \log\Delta_{\Omega;M_A}
      \bigr)
    }.
    \label{eq:hsmi-positive-generator}
\end{equation}
Here the difference is defined on
\(\Dom(\log\Delta_{\Omega;M_B})\cap
\Dom(\log\Delta_{\Omega;M_A})\), where it is essentially self-adjoint.
The operator \(P\) is positive and self-adjoint.  The strongly continuous
unitary group
\begin{equation}
    U_b\coloneqq e^{-ibP},
    \qquad b\in\mathbb R,
    \label{eq:hsmi-unitaries}
\end{equation}
fixes the distinguished vector, \(U_b\Omega=\Omega\), and generates the
family
\begin{equation}
    M_b\coloneqq U_{-b}M_AU_b,
    \qquad b\in\mathbb R.
    \label{eq:hsmi-algebra-family}
\end{equation}
The family of algebras satisfies
\begin{equation}
    M_0=M_A,
    \qquad M_1=M_B,
    \qquad
    M_c\subset M_b\quad( c> b\geq 0).
    \label{eq:hsmi-nesting}
\end{equation}
The unitary group satisfies the Borchers covariance relation,
\begin{equation}
    \Delta_{\Omega;M_A}^{is}
    U_b
    \Delta_{\Omega;M_A}^{-is}
    =
    U_{e^{-2\pi s}b},
    \qquad s,b\in\mathbb R.
    \label{eq:borchers-covariance}
\end{equation}
If \(J_{\Omega;M_A}\) denotes the modular conjugation of
\((M_A,\Omega)\), then
\begin{equation}
    J_{\Omega;M_A}U_bJ_{\Omega;M_A}=U_{-b}.
    \label{eq:modular-conjugation-null-translations}
\end{equation}
\end{theorem}

A concrete geometric realization of HSMI is provided by null-translated
Rindler wedges. Take constant cuts $A(\vec y)=0$ and
$B(\vec y)=\ell$, with $\ell>0$, and write
\[
    R_a=\{x^+>a,\ x^-<0\}.
\]
Let $V(a)=e^{iaP_+}$ implement the spacetime translation
$x^+\mapsto x^++a$, so that Poincar\'e covariance gives
\[
    V(a)\mathcal A(R_0)V(a)^*=\mathcal A(R_a).
\]
The spectrum condition implies $P_+\geq0$, and the
Bisognano--Wichmann property identifies vacuum modular flow
with the wedge-preserving boosts \cite{Bisognano:1975ih}.
With our conventions, these boosts give
\[
    \Delta_{\Omega;M_A}^{-is}M_B
    \Delta_{\Omega;M_A}^{is}
    =\mathcal A(R_{e^{2\pi s}\ell})
    \subseteq M_B,
    \qquad s>0.
\]
Thus $(M_B\subset M_A,\Omega)$ is a half-sided modular
inclusion. On the common core of smooth vectors for the
Poincar\'e representation, the translated boost generators
satisfy
\[
    \log\Delta_{\Omega;M_B}
    -\log\Delta_{\Omega;M_A}
    =2\pi\ell P_+.
\]
Taking the closure in \eqref{eq:hsmi-positive-generator}
therefore gives $P=\ell P_+$ and $U_b=V(-b\ell)$.
Consequently,
\[
    U_{-c}M_AU_c
    =\mathcal A(R_{c\ell})
    =M_{A_c},
    \qquad c\geq0.
\]
The geometric identification required in
Corollary~\ref{cor:integer-renyi-qnec-fixed} therefore holds
for this family, including in interacting theories satisfying
the stated assumptions.

For general transverse profiles, identifying the geometric
family $M_{A_c}$ with the canonical HSMI family requires
additional null-cut modular-flow input, as in the settings
covered by \cite{Casini:2017roe}.
Therefore the identification
\[
    M_b=M_{A_b}, \qquad b\geq0, \qquad \text{ with } A_b=A+b(B-A),
\]
is retained as an explicit hypothesis.
When the stress-tensor
representation is available, the canonical generator
\eqref{eq:hsmi-positive-generator} agrees with
\eqref{eq:weighted-ane-generator} under this identification.

Independently of any geometric
realization, Borchers covariance implies
\begin{equation}
    \Delta_{\Omega;M_A}^{is}
    M_b
    \Delta_{\Omega;M_A}^{-is}
    =M_{b e^{-2\pi s}},
    \label{eq:modular-rescaling-of-cuts}
\end{equation}
which is the abstract counterpart of the Rindler dilation of affine
horizon distance. 
We henceforth reserve \(M_c\) for this canonical HSMI family. 
Setting \(M\coloneqq M_A\), define
\begin{equation}
    T_c:M\longrightarrow M,
    \qquad
    T_c(a)\coloneqq U_{-c}aU_c,
    \qquad c\geq0.
    \label{eq:def-Tc}
\end{equation}
Although \(T_c=\operatorname{Ad}(U_{-c})\) is an automorphism of
\(B(\mathcal H)\), its restriction to \(M\) has range
\(T_c(M)=M_c\subset M\).
Thus \(T_c\) is a normal unital injective \(*\)-endomorphism of the fixed
algebra \(M\).  The family \((T_c)_{c\geq0}\) satisfies
\begin{equation}
    T_{c+d}=T_cT_d,
    \qquad
    \omega\circ T_c=\omega,
    \qquad c,d\geq0,
    \label{eq:Tc-semigroup-and-invariance}
\end{equation}
and \(T_c(a)\to a\) in the strong-\(*\) topology as \(c\downarrow0\).

Let \(\sigma_s^\omega(a)\coloneqq
\Delta_{\Omega;M}^{is}a\Delta_{\Omega;M}^{-is}\) be the modular
automorphism group of \((M,\omega)\).  In terms of \(T_c\), Borchers
covariance becomes
\begin{align}
    \sigma_s^\omega\circ T_c
    &=T_{c e^{-2\pi s}}\circ\sigma_s^\omega,
    \label{eq:Tc-modular-covariance}
    % \\
    % T_c\circ\sigma_s^\omega
    % &=\sigma_s^\omega\circ T_{c e^{2\pi s}}.
    % \label{eq:Tc-modular-covariance-2}
\end{align}
State preservation in \eqref{eq:Tc-semigroup-and-invariance} gives the
\(L^p\)-semigroup used below.  We will use Borchers covariance
\eqref{eq:Tc-modular-covariance} to derive the cyclic Ward identity in Section~\ref{sec:integer-renyi-qnec}.

%%%%%%%%%%%%%%%%%%%%%%%%%%%%%%%%%%%%%%%%%%%%%%%%%%%%%%%%%%%%%%%%%%%%%%%%%%%%%%
\subsection{Sandwiched R\'{e}nyi divergence on a fixed algebra}
\label{subsec:fixed-algebra-SRD}

We formulate the SRD on the fixed algebra \(M\equiv M_0\).  The maps \(T_c\) will then identify the divergence on each cut algebra with a divergence on \(M\).

In finite dimensions, if \(\rho\) and \(\sigma\) are density matrices of normalized states and
\(\sigma>0\), their sandwiched R\'{e}nyi divergence for \(\alpha>1\) is\footnote{SRD is also well-defined for \(\alpha\in(0,1)\), and satisfies DPI for \(\alpha\in[1/2,1)\). However, since RQNEC is known to be violated in this lower-order region \cite{Moosa:2020jwt}, we do not introduce it.}
\cite{Muller-Lennert:2013liu,Wilde:2013bdg}
\begin{equation}
    D_\alpha(\rho\Vert\sigma)
    =
    \frac{1}{\alpha-1}
    \log
    \operatorname{Tr}\!\left[
      \left(
        \sigma^{\frac{1-\alpha}{2\alpha}}
        \rho
        \sigma^{\frac{1-\alpha}{2\alpha}}
      \right)^\alpha
    \right],
    \qquad \alpha>1.
    \label{eq:finite-dimensional-SRD}
\end{equation}
SRD can be defined on arbitrary von Neumann algebras
\cite{Berta:2016vnw,Jencova:2016tqz}. We work with the formulation of 
\cite{Jencova:2016tqz}, based on Haagerup's noncommutative \(L^p\) spaces
\cite{Terp:1981lp,Hiai:2021qfv,Hiai:2021cks}.  The continuous core
\(\mathcal R=M\rtimes_{\sigma^\omega}\mathbb R\) has a canonical
faithful normal semifinite trace \(\tau\) and dual action \(\theta\),
with \(\tau\circ\theta_s=e^{-s}\tau\).  For \(1\leq p<\infty\),
\(L^p(M)\) consists of the \(\tau\)-measurable operators affiliated with
\(\mathcal R\) that satisfy
\begin{equation}
    \theta_s(x)=e^{-s/p}x,\qquad s\in\mathbb R.
    \label{eq:Haagerup-homogeneity}
\end{equation}
We set \(L^\infty(M)=M\).

The canonical identification \(M_*\cong L^1(M)\) assigns to each normal
functional \(\varphi\) a density \(h_\varphi\), positive when
\(\varphi\) is positive, via
\begin{equation}
    \operatorname{tr}(h_\varphi a)=\varphi(a),\qquad a\in M.
    \label{eq:Haagerup-density-pairing}
\end{equation}
Here \(\operatorname{tr}\) is the canonical linear functional on
\(L^1(M)\), distinct from \(\tau\).  It determines the $L^p$-norm
\begin{equation}
    \|x\|_p=\bigl(\operatorname{tr}|x|^p\bigr)^{1/p}.
    \label{eq:Haagerup-Lp-norm}
\end{equation}
Products are understood as closed products of measurable operators.  The H\"older and cyclicity results used below are recalled in
Appendix~\ref{appsec:Haagerup-analytic-facts}.

For \(\alpha>1\), set
\[
    \eta_\alpha\coloneqq\frac{\alpha-1}{2\alpha}.
\]
Since
\(h_\omega^{\eta_\alpha}\in
L^{1/\eta_\alpha}(M)\), the generalized H\"older inequality gives
\[
    x\in L^\alpha(M) 
    \quad\Longrightarrow\quad
    h_\omega^{\eta_\alpha}x h_\omega^{\eta_\alpha}
    \in L^1(M).
\]

\begin{definition}[Sandwiched R\'{e}nyi divergence and density]
\label{def:SRD}
Let \(0\neq\psi\in M_*^+\).  If there exists an element
\(x\in L^\alpha(M)_+\) satisfying
\begin{equation}
    h_\psi
    =
    h_\omega^{\eta_\alpha}
    x
    h_\omega^{\eta_\alpha},
    \label{eq:x-op-defn-SRD}
\end{equation}
then \(x\) is unique \cite[Lemma~8]{Kato:2023aro}.  We denote it by
\(x_{\psi\mid\omega,\alpha}\) and call it the sandwiched
\(L^\alpha\)-density of \(\psi\) relative to \(\omega\).  Define
\begin{align}
    Q_\alpha(\psi\Vert\omega;M)
    &\coloneqq
    \begin{cases}
      \operatorname{tr}\!\left(
        x_{\psi\mid\omega,\alpha}^{\,\alpha}
      \right)
      =\|x_{\psi\mid\omega,\alpha}\|_\alpha^\alpha,
      &\text{if \eqref{eq:x-op-defn-SRD} has a solution in
       \(L^\alpha(M)_+\),}\\[0.3em]
      +\infty,
      &\text{otherwise},
    \end{cases}
    \label{eq:QSRD}\\
    D_\alpha(\psi\Vert\omega;M)
    &\coloneqq
    \frac{1}{\alpha-1}
    \log\frac{Q_\alpha(\psi\Vert\omega;M)}{\psi(1)}.
    \label{eq:SRD-definition}
\end{align}
When \(\omega\) and \(\alpha\) are fixed, we write \(x_\psi\) for
\(x_{\psi\mid\omega,\alpha}\).
\end{definition}

For a normalized state, \(\psi(1)=1\).  For non-zero \(\psi\in M_*^+\), the denominator in \eqref{eq:SRD-definition} specifies the
normalized extension we use.  Under the standard identification of the
Haagerup spaces of a finite-dimensional matrix algebra with the corresponding
Schatten spaces,
\(h_\psi\) and \(h_\omega\) correspond respectively to \(\rho\) and \(\sigma\) in \eqref{eq:finite-dimensional-SRD}, and
\eqref{eq:x-op-defn-SRD} becomes
\begin{equation}
    x_{\psi\mid\omega,\alpha}
    =
    \sigma^{\frac{1-\alpha}{2\alpha}}
    \rho
    \sigma^{\frac{1-\alpha}{2\alpha}},
    \label{eq:matrix-sandwiched-density}
\end{equation}
recovering \eqref{eq:finite-dimensional-SRD} from Definition~\ref{def:SRD}.

Following~\cite{Hollands:2025glm}, we transfer the problem from the
moving algebras \(M_c\) to the fixed algebra \(M\).  Recall that
\(T_c:M\to M_c\subset M\), defined in \eqref{eq:def-Tc}, is a normal
unital \(*\)-isomorphism.  For \(0\neq\psi\in M_*^+\), let
\begin{equation}
    \psi_c
    \coloneqq
    (\psi|_{M_c})\circ T_c
    =\psi\circ T_c.
    \label{eq:null-translated-functional}
\end{equation}
Thus \(\psi_c\) is the pullback to \(M\) of the restriction of \(\psi\)
to \(M_c\).  If \(\psi(a)=\langle\Psi,a\Psi\rangle\) for some
\(\Psi\in\mathcal H\), then
\begin{equation}
    \psi_c(a)
    =
    \langle U_c\Psi,aU_c\Psi\rangle,
    \qquad a\in M,
    \label{eq:translated-vector-state}
\end{equation}
so varying the algebra is equivalent to varying the vector while keeping
the algebra fixed.

\begin{lemma}[Fixed-algebra formulation]
\label{lem:fixedalgebraSRD}
Let \(\alpha>1\), \(c\geq0\), and \(0\neq\psi\in M_*^+\).  Then
\begin{equation}
    S_\alpha(c)
    \coloneqq
    D_\alpha\!\left(
      \psi|_{M_c}\Vert\omega|_{M_c};M_c
    \right)
    =
    D_\alpha(\psi_c\Vert\omega;M).
    \label{eq:fixed-algebra-SRD}
\end{equation}
\end{lemma}
We abbreviate \(D_\alpha\!\left(\psi|_{M_c}\Vert\omega|_{M_c};M_c\right)\) by
\(D_\alpha(\psi\Vert\omega;M_c)\).
\begin{proof}
The SRD is invariant under normal \(*\)-isomorphisms, by data processing
applied to the isomorphism and its inverse
\cite[Theorem~3.14]{Jencova:2016tqz}.  Apply this to \(T_c:M\to M_c\),
using \((\psi|_{M_c})\circ T_c=\psi_c\) and
\((\omega|_{M_c})\circ T_c=\omega\).
\end{proof}

For \(c_2\geq c_1\), the inclusion \(M_{c_2}\subset M_{c_1}\) and data
processing give
\begin{equation}
    S_\alpha(c_2)\leq S_\alpha(c_1).
    \label{eq:SRD-null-monotonicity}
\end{equation}
Thus finiteness at the initial cut implies finiteness at every later cut. 
% To prove
% convexity, we will study the corresponding sandwiched densities on the
% fixed algebra.

%%%%%%%%%%%%%%%%%%%%%%%%%%%%%%%%%%%%%%%%%%%%%%%%%%%%%%%%%%%%%%%%%%%%%%%%%%%%%%
\subsection{Main results}
\label{subsec:main-results}

We now state the abstract operator-algebraic
log-convexity theorem, followed by the geometric RQNEC which follows as a corollary. 

The
state-preserving endomorphisms \(T_c\) induce a semigroup on the Haagerup \(L^p\) densities.  
For \(1\leq p<\infty\), consider the
symmetric embedding
\begin{equation}
    j_p(a)
    \coloneqq
    h_\omega^{1/(2p)}a h_\omega^{1/(2p)},
    \qquad a\in M,
    \label{eq:main-symmetric-embedding}
\end{equation}
whose range is dense in \(L^p(M)\)
\cite{KOSAKI198429,Hiai:2021cks}.  The extension
theorem for state-preserving positive maps
\cite[Theorem~5.1]{Haagerup:2008arm} gives a unique positive
contraction \(V_c^{(p)}:L^p(M)\to L^p(M)\), characterized by
\begin{equation}
    V_c^{(p)}j_p(a)=j_p(T_c(a)).
    \label{eq:main-Lp-extension}
\end{equation}
These extensions form the strongly continuous contraction
semigroup constructed in
Proposition~\ref{prop:null-translation-Ln-semigroup}. For integer \(p=n\geq2\), we establish the following theorem in Section \ref{subsec:smooth-orbit-log-convexity-fixed}.

\begin{theorem}[Log-convexity under half-sided modular translations]
\label{thm:general-log-convexity}
Let \((M_B\subset M_A,\Omega)\) be a half-sided modular inclusion, set
\(M=M_A\), and let \(n\geq2\) be an integer.  For every non-zero
\(x_0\in L^n(M)_+\), set
\begin{equation}
    x_c\coloneqq V_c^{(n)}x_0,
    \qquad
    Q_n(c)\coloneqq\operatorname{tr}(x_c^n)
                   =\|x_c\|_n^n,
    \qquad c\geq0.
    \label{eq:main-Qn}
\end{equation}
Then $Q_n$ is strictly positive and continuous on $[0,\infty)$,
real analytic on $(0,\infty)$, and log-convex on $[0,\infty)$:
\begin{equation}\label{eq:regular-orbit-log-convexity}
    Q_n\bigl((1-\theta)c_0+\theta c_1\bigr)
    \leq Q_n(c_0)^{1-\theta}Q_n(c_1)^\theta,
    \qquad c_0,c_1\geq0,\quad 0\leq\theta\leq1.
\end{equation}
\end{theorem}

The real analyticity of $Q_n$ on $(0,\infty)$ allows writing the
log-convexity statement as the pointwise differential statement
\begin{equation}
    Q_n(c)Q_n''(c)-Q_n'(c)^2\geq0,
    \qquad c>0.
\end{equation}

Lemma~\ref{lem:positive-Ln-sandwiched-density} shows that
every non-zero \(x_0\in L^n(M)_+\) determines a unique non-zero
\(\psi\in M_*^+\) through
\(h_\psi=h_\omega^{\eta_n}x_0h_\omega^{\eta_n}\).  Thus the elements
covered by Theorem~\ref{thm:general-log-convexity} correspond exactly to
the positive functionals with finite order-\(n\) SRD (see Definition~\ref{def:SRD}).

\begin{corollary}[Integer R\'{e}nyi QNEC]
\label{cor:integer-renyi-qnec-fixed}
Suppose that \((M_B\subset M_A,\Omega)\) is a half-sided modular inclusion
and that its canonical translated algebras coincide with the geometric
null-cut family,
\begin{equation}
    U_{-c}M_AU_c=M_{A_c},
    \qquad A_c=A+c(B-A),
    \qquad c\geq0.
    \label{eq:geometric-canonical-family-identification}
\end{equation}
Let \(\omega\) be the faithful normal vacuum state induced by \(\Omega\), and set
\(M=M_A\).  Let \(n\geq2\) be an integer.  If \(0\neq\psi\in M_*^+\) satisfies
\begin{equation}
    D_n(\psi\Vert\omega;M)<\infty,
    \label{eq:main-finite-SRD}
\end{equation}
then the geometric and canonical sandwiched R\'{e}nyi functions agree:
\begin{equation}
\begin{aligned}
    S_n^{\mathrm{geom}}(c)
    &=S_n(c)
    =D_n(\psi\Vert\omega;M_{A_c})\\
    &=D_n(\psi\Vert\omega;M_c)
    =D_n(\psi_c\Vert\omega;M),
    \qquad c\geq0.
\end{aligned}
    \label{eq:main-Sn}
\end{equation}
This common function is finite, continuous, non-increasing, and convex on
\([0,\infty)\), and real analytic on $(0,\infty)$.  
%Equivalently, \(S_n''\) is a non-negative distribution on \((0,\infty)\).  
In particular, for every positive cut $c>0$,
\begin{equation}
    S_n''(c)\geq0.
    \label{eq:main-RQNEC-pointwise}
\end{equation}
Moreover, either the sandwiched-density orbit is stationary and $S_n$
is constant, or
\[
    S_n'(c)<0,
    \qquad S_n''(c)>0,
    \qquad c>0.
\]
\end{corollary}

\begin{proof}
Put \(\eta_n=(n-1)/(2n)\) and let
\(
    x_0\coloneqq x_{\psi\mid\omega,n}.
\)
By Definition~\ref{def:SRD}, the finiteness hypothesis
\eqref{eq:main-finite-SRD} gives the unique element
\(x_0\in L^n(M)_+\) satisfying
\begin{equation}
    h_\psi=h_\omega^{\eta_n}x_0h_\omega^{\eta_n},
    \qquad
    Q_n(\psi\Vert\omega;M)=\operatorname{tr}(x_0^n)<\infty.
    \label{eq:main-initial-sandwiched-density}
\end{equation}
Let \(x_c=V_c^{(n)}x_0\).  The density-propagation result,
Lemma~\ref{lem:density-propagation-fixed}, shows that
\begin{equation}
    h_{\psi_c}=h_\omega^{\eta_n}x_ch_\omega^{\eta_n},
    \qquad
    x_c=x_{\psi_c\mid\omega,n}.
    \label{eq:main-density-propagation-preview}
\end{equation}
Since \(T_c\) is unital, \(\psi_c(1)=\psi(1)\).  Hence
Lemma~\ref{lem:fixedalgebraSRD} and Definition~\ref{def:SRD} give
\begin{equation}
    D_n(\psi\Vert\omega;M_c)
    =D_n(\psi_c\Vert\omega;M)
    =\frac{1}{n-1}
      \log\!\left(
        \frac{\operatorname{tr}\!\left[
          \bigl(V_c^{(n)}x_0\bigr)^n
        \right]}{\psi(1)}
      \right).
    \label{eq:main-four-level-dictionary}
\end{equation}
By Theorem~\ref{thm:general-log-convexity}, the last expression is convex.
Its continuity follows from the strong continuity of \(V_c^{(n)}\) and the
identity \(Q_n(c)=\|V_c^{(n)}x_0\|_n^n\), while finiteness and monotonicity
follow from \eqref{eq:SRD-null-monotonicity}. The alternative
between stationarity and strict decrease with strict convexity follows
from Proposition~\ref{prop:rqnec-saturation}.
\end{proof}

We emphasize that \eqref{eq:main-finite-SRD} is the only
state-dependent hypothesis.  In particular, no differentiability of
\(c\mapsto\psi_c\) or generator-domain condition on its sandwiched density
is imposed.

At order \(n=2\), Section~\ref{subsec:rqnec-spectral} gives a direct
spectral proof of these results.  Under the standard-form identification
\(L^2(M)\simeq\mathcal H\), the semigroup is \(e^{-cP}\), so \(Q_2\) is
a Laplace transform of a positive measure.  For general integer \(n\),
we instead use the cyclic structure of the \(n\)-fold product in \(\operatorname{tr} (x_c^n)\).  
This is where the integer restriction
enters the argument.

Our second main result,
Theorem~\ref{thm:integer-escort-ant}, gives a variational
energy formula for the first null derivative of the SRD.
At positive cuts, it requires only finite initial
divergence for even orders, with an additional domain
condition for odd orders.

%%%%%%%%%%%%%%%%%%%%%%%%%%%%%%%%%%%%%%%%%%%%%%%%%%%%%%%%%%%%%%%%%%%%%%%%%%%%%%%%%%%%

\section{\texorpdfstring{$L^p$}{Lp} dynamics of a half-sided modular
inclusion}
\label{sec:Ln-dynamics}

The state-preserving endomorphisms $T_c$ introduced above extend to
positive contractions $V_c^{(p)}$. In this section, we establish their
strong continuity, complementary-exponent duality, modular covariance,
and action on sandwiched densities. We also identify the standard-form
realization of $V^{(2)}$. Together with the interpolation argument in
Lemma~\ref{lem:null-semigroup-analytic}, this gives the interior
regularity needed to differentiate the integer moments in Section~4.

Throughout this section, \((M_B\subset M_A,\Omega)\) is the half-sided
modular inclusion of Section~\ref{subsec:HSMI-null-translations},
\(M=M_A\), \(h=h_\omega\). We use
\(p\) to denote a real number in \([1,\infty)\) and \(n\) to denote 
an integer greater than 1. The strong-continuity, duality, covariance and 
transport results
below hold for all $1\le p<\infty$. Analyticity and the resulting
interior regularity hold for $1<p<\infty$.  We write
\begin{equation}
    q\coloneqq\frac{p}{p-1},
    \qquad
    \eta_p\coloneqq\frac{1}{2q}
                    =\frac{p-1}{2p}.
    \label{eq:Ln-dynamics-exponents}
\end{equation}
Thus \(L^{q}(M)\) is the Banach dual of \(L^p(M)\) under the
Haagerup pairing \cite[Theorem~9.27]{Hiai:2021cks}
\begin{equation}
    \langle x,y\rangle_{p,q}
    \coloneqq\operatorname{tr}(xy),
    \qquad
    x\in L^p(M),\quad y\in L^{q}(M).
    \label{eq:Ln-Lqn-pairing}
\end{equation}
At $p=1$, we take $q=\infty$, with
$j_\infty(a)=a$ and $V_c^{(\infty)}=T_c$.
% We take Hilbert-space inner products to be linear in the second variable.

%%%%%%%%%%%%%%%%%%%%%%%%%%%%%%%%%%%%%%%%%%%%%%%%%%%%%%%%%%%%%%%%%%%%%%%%%%%%%%
\subsection{The null-translation semigroup and its dual}
\label{subsec:Ln-translation-semigroup}

The contractions \(V_c^{(p)}\) inherit the continuous semigroup structure
of the null translations.

\begin{proposition}[Null translations on \(L^p(M)\)]
\label{prop:null-translation-Ln-semigroup}
With \(V_c^{(p)}\) defined in \eqref{eq:main-Lp-extension}, the family
\[
    V^{(p)}=(V_c^{(p)})_{c\geq0}
\]
is a strongly continuous semigroup of positive contractions on
\(L^p(M)\).  In particular, for \(c_0,c,d\geq0, \ x\in L^p(M)\),
\begin{equation}\label{eq:Ln-translation-semigroup}
    V_{c+d}^{(p)}=V_c^{(p)}V_d^{(p)},
    \qquad
    V_0^{(p)}=I,
    \qquad
    \lim_{c\to c_0}\|V_c^{(p)}x-V_{c_0}^{(p)}x\|_p=0.
\end{equation}
\end{proposition}

\begin{proof}
As mentioned above \eqref{eq:main-Lp-extension}, 
\cite[Theorem~5.1]{Haagerup:2008arm} gives positive contractions satisfying
\[
    V_c^{(p)}j_p(a)=j_p(T_c(a)),\qquad a\in M.
\]
The semigroup identities follow on \(j_p(M)\) from \(T_{c+d}=T_cT_d\)
and \(T_0=I\), and extend to \(L^p(M)\) by contractivity and density.

For continuity, the unitary implementation gives
\(T_c(a)\to a\) strongly as \(c\downarrow0\), with
\(\|T_c(a)\|=\|a\|\).  Lemma~\ref{lem:sandwich-continuity}
therefore implies
\[
    \|V_c^{(p)}j_p(a)-j_p(a)\|_p
    =\|j_p(T_c(a))-j_p(a)\|_p\longrightarrow0.
\]
Contractivity extends this convergence from the dense subspace
\(j_p(M)\) to all of \(L^p(M)\).  The semigroup property then gives
continuity at every \(c_0\geq0\).
\end{proof}

Let \(G_p\) denote the generator of \(V^{(p)}\):
\begin{equation}
    G_px
    \coloneqq
    \lim_{c\downarrow0}\frac{V_c^{(p)}x-x}{c},
    \qquad x\in\Dom(G_p),
    \label{eq:Gn-definition}
\end{equation}
where the limit is taken in \(L^p(M)\).  As usual for a strongly
continuous semigroup, \(G_p\) is closed and densely defined.  We will also
use the standard orbit-regularity statement
\begin{equation}
\begin{aligned}
    V_c^{(p)}\Dom(G_p^k)&\subset\Dom(G_p^k),\\
    \frac{d^j}{dc^j}V_c^{(p)}x
    &=V_c^{(p)}G_p^jx
     =G_p^jV_c^{(p)}x,
    \qquad 0\leq j\leq k,
\end{aligned}
    \label{eq:semigroup-orbit-regularity}
\end{equation}
for \(k\geq1\), \(x\in\Dom(G_p^k)\), and \(c\geq0\), with
right derivatives at \(c=0\).  For these facts from
\(C_0\)-semigroup theory, see
\cite[Lemma~II.1.3(ii) and Theorem~II.1.4]
{engel.nagel:00:one-parameter}.

\begin{lemma}[Standard-form implementation at order two]
\label{lem:null-l2-implementation}
The rule
\begin{equation}
    W(ah^{1/2})\coloneqq a\Omega,
    \qquad a\in M,
    \label{eq:n2-standard-form-unitary}
\end{equation}
defines a unitary \(W:L^2(M)\to\mathcal H\).  Under this unitary,
\begin{equation}
    W V_c^{(2)} W^*=e^{-cP},
    \qquad c\geq0.
    \label{eq:n2-semigroup-implementation}
\end{equation}
\end{lemma}

\begin{proof}
For \(a\in M\),
\[
    \|ah^{1/2}\|_2^2
    =\omega(a^*a)
    =\|a\Omega\|_{\mathcal H}^2.
\]
The subspaces \(Mh^{1/2}\subset L^2(M)\) and
\(M\Omega\subset\mathcal H\) are dense, so \(W\) extends to a unitary.
On these two dense subspaces it intertwines the Tomita operators, and hence
their polar decompositions.  The standard-form realization of Haagerup
\(L^2(M)\) therefore gives
\begin{equation}
    Wj_2(a)=\Delta_{\Omega;M}^{1/4}a\Omega,
    \qquad a\in M;
    \label{eq:n2-symmetric-vector-identification}
\end{equation}
see \cite[Theorem~9.29 and Proposition~10.6]{Hiai:2021cks}.

We use Borchers covariance at the imaginary modular parameter \(s=-i/4\).
Continuation to the quarter strip gives the domain inclusion
\(U_{-c}a\Omega\in\Dom(\Delta_{\Omega;M}^{1/4})\) and
\begin{equation}
    \Delta_{\Omega;M}^{1/4}U_{-c}a\Omega
    =e^{-cP}\Delta_{\Omega;M}^{1/4}a\Omega,
    \qquad a\in M,\quad c\geq0.
    \label{eq:n2-imaginary-Borchers}
\end{equation}
This follows by setting \(s=-i/4\) in the strip-analytic relation associated
with \eqref{eq:borchers-covariance}.  With our convention
\(U_b=e^{-ibP}\), the corresponding complex translation parameter is
\(-ic\)
\cite[Theorem~2.1(1),(8), Lemma~2.11(1), and
Theorem~2.12(8)]{Araki:2005we}.

Using \(T_c(a)\Omega=U_{-c}a\Omega\), we therefore obtain
\begin{align*}
    WV_c^{(2)}j_2(a)
    &=Wj_2(T_c(a))\\
    &=\Delta_{\Omega;M}^{1/4}U_{-c}a\Omega\\
    &=e^{-cP}\Delta_{\Omega;M}^{1/4}a\Omega
     =e^{-cP}Wj_2(a).
\end{align*}
The density of \(j_2(M)\) in \(L^2(M)\), together with boundedness of
both operators, proves \eqref{eq:n2-semigroup-implementation}.
\end{proof}

% Lemma~\ref{lem:null-semigroup-analytic} combines this implementation
% with interpolation to show that $V^{(p)}$ is holomorphic for
% $1<p<\infty$. In particular, for every $x\in L^p(M)$,
% \begin{equation}
% \label{eq:null-interior-derivatives}
%     V_c^{(p)}x\in\bigcap_{k\geq1}\operatorname{Dom}(G_p^k),
%     \qquad
%     \frac{d^k}{dc^k}V_c^{(p)}x=G_p^kV_c^{(p)}x,
%     \qquad c>0,\quad k\geq1.
% \end{equation}
% These are interior statements. Derivatives at $c=0$ still require the
% generator-domain hypotheses in the preceding orbit-regularity formula.

\begin{lemma}[Interior regularity of the null semigroup]
\label{lem:null-semigroup-analytic}
For every $1<p<\infty$, the semigroup $V^{(p)}$ is holomorphic.
For each $x\in L^p(M)$, its orbit is norm analytic on $(0,\infty)$,
and
\[
    V_c^{(p)}x\in\bigcap_{k\geq1}\operatorname{Dom}(G_p^k),
    \qquad
    \frac{d^k}{dc^k}V_c^{(p)}x=G_p^kV_c^{(p)}x,
    \qquad c>0,\quad k\geq1.
\]
\end{lemma}

\begin{proof}
Use the symmetric Kosaki realization
\[
    \iota_p(x):=h^{\eta_p}xh^{\eta_p},
    \qquad \eta_p=\frac{p-1}{2p},
    \qquad X_p:=\iota_p(L^p(M)),
\]
with the norm transported from $L^p(M)$. Set $X_1=L^1(M)$ and
\[
    X_\infty:=h^{1/2}Mh^{1/2},
    \qquad
    \|h^{1/2}ah^{1/2}\|_{X_\infty}:=\|a\|.
\]
These spaces form a compatible interpolation scale inside $L^1(M)$.
In particular, from \cite[Section 4]{KOSAKI198429}
\[
    X_p=[X_\infty,X_2]_{2/p},
    \qquad 2<p<\infty,
\]
and
\[
    X_p=[X_2,X_1]_{2/p-1},
    \qquad 1<p<2.
\]

On the symmetric cores,
\[
    \iota_pj_p(a)=j_1(a),
    \qquad
    \iota_pV_t^{(p)}j_p(a)=V_t^{(1)}\iota_pj_p(a).
\]
Density and continuity therefore give
\[
    \iota_pV_t^{(p)}=V_t^{(1)}\iota_p.
\]
On $X_\infty$, the same operator acts by
\[
    h^{1/2}ah^{1/2}
    \longmapsto h^{1/2}T_t(a)h^{1/2}.
\]
It is a contraction in the $X_\infty$ norm because $T_t$ is
contractive. Thus the operators are consistent on the interpolation
scale, including its endpoints.

Lemma~\ref{lem:null-l2-implementation} and the positivity of $P$ give
\[
    \|I-V_t^{(2)}\|=\|I-e^{-tP}\|\leq1.
\]
On $X_1$, the triangle inequality and contractivity give
\[
    \|I-V_t^{(1)}\|
    \leq\|I\|+\|V_t^{(1)}\|
    \leq2.
\]
The transported action on $X_\infty$ has the same bound for its
difference from the identity. Interpolation therefore yields
\[
    \|I-V_t^{(p)}\|
    \leq
    \begin{cases}
        2^{1-2/p},&2<p<\infty,\\
        2^{2/p-1},&1<p<2.
    \end{cases}
\]
Together with the order-two estimate, this gives
\[
    \sup_{t>0}\|I-V_t^{(p)}\|
    \leq2^{|1-2/p|}<2,
    \qquad 1<p<\infty.
\]
Kato's criterion implies holomorphy; see \cite[Corollary~2.5(a)]{Fackler:2012}. The domain inclusions and derivative formulas are
the standard interior regularity properties of a holomorphic
$C_0$-semigroup; see \cite[Section~II.4]{engel.nagel:00:one-parameter}.
\end{proof}

For a general positive state-preserving map $\widetilde T$ on a von Neumann algebra, let 
$\widetilde V^{(p)}$ be its $L^p$-extension from \cite[Theorem~5.1]{Haagerup:2008arm}. The Banach adjoint of $\widetilde V^{(p)}$ need not coincide with $\widetilde V^{(q)}$, the $L^q$-extension of $\widetilde T$ \cite[Proposition~5.4]{Haagerup:2008arm}.
However, the null
translations of an HSMI satisfy the reflection
relation \eqref{eq:modular-conjugation-null-translations}, which gives
duality between their extensions at complementary exponents $p$ and $q$.

\begin{proposition}[Complementary-exponent duality]
\label{prop:equal-dual-semigroups}
Let $1\leq p<\infty$, and $q$ be such that $1/p+1/q=1$. For every \(c\geq0\),
\begin{equation}
    \operatorname{tr}\!\left(
      (V_c^{(p)}x)y
    \right)
    =
    \operatorname{tr}\!\left(
      x(V_c^{(q)}y)
    \right),
    \qquad
    x\in L^p(M),\quad y\in L^{q}(M).
    \label{eq:Ln-translation-duality}
\end{equation}
\end{proposition}

\begin{proof}
It is enough to verify \eqref{eq:Ln-translation-duality} for
\(x=j_p(a)\) and \(y=j_{q}(b)\), with \(a,b\in M\).  Write
\(J=J_{\Omega;M}\).  The symmetric-core pairing
of Lemma~\ref{lem:complementary-core-pairing} gives
\begin{align}
    \operatorname{tr}\!\left(
      V_c^{(p)}j_p(a)\,j_{q}(b)
    \right)
    &=
    \left\langle
      T_c(a)^*\Omega,Jb^*\Omega
    \right\rangle_{\mathcal H}
    \notag\\
    &=
    \left\langle
      a^*\Omega,U_cJb^*\Omega
    \right\rangle_{\mathcal H}.
    \label{eq:Ln-duality-standard-form-step}
\end{align}
Here we used \(U_c\Omega=\Omega\), so that
\(T_c(a)^*\Omega=U_{-c}a^*\Omega\).
The modular-reflection identity
\(JU_{-c}J=U_c\) implies
\[
    U_cJb^*\Omega
    =
    J\,T_c(b)^*\Omega.
\]
Applying the core-pairing formula once more to
\eqref{eq:Ln-duality-standard-form-step} gives
\[
    \operatorname{tr}\!\left(
      V_c^{(p)}j_p(a)\,j_{q}(b)
    \right)
    =
    \operatorname{tr}\!\left(
      j_p(a)V_c^{(q)}j_{q}(b)
    \right).
\]
Density of the two symmetric cores and H\"older's inequality extend the
identity to arbitrary \(x\in L^p(M)\) and \(y\in L^{q}(M)\).
\end{proof}

We next lift Borchers covariance from \(M\) to \(L^p(M)\).  For
\(s\in\mathbb R\), let \(\Sigma_s^{(p)}\) be the $L^p$-extension~\cite[Theorem~5.1]{Haagerup:2008arm} of the modular
automorphism \(\sigma_s^\omega\), characterized by
\begin{equation}
    \Sigma_s^{(p)}j_p(a)
    =
    j_p(\sigma_s^\omega(a)),
    \qquad a\in M.
    \label{eq:Ln-modular-extension}
\end{equation}

\begin{proposition}[Borchers covariance on \(L^p(M)\)]
\label{prop:Ln-modular-covariance}
Let $1\leq p<\infty$. The family
\((\Sigma_s^{(p)})_{s\in\mathbb R}\) is a strongly continuous group of
positive isometries on \(L^p(M)\).  Moreover, for \(s\in\mathbb{R},\,c\geq0,\)
\begin{align}
    \Sigma_s^{(p)}V_c^{(p)}
    &=
    V_{ce^{-2\pi s}}^{(p)}\Sigma_s^{(p)},
    \label{eq:Ln-Borchers-covariance}\\
    \Sigma_s^{(p)}\Dom(G_p)
    &=
    \Dom(G_p),
    \notag\\
    G_p\Sigma_s^{(p)}x
    &=
    e^{2\pi s}\Sigma_s^{(p)}G_px,
    \qquad x\in\Dom(G_p).
    \label{eq:Ln-generator-modular-covariance}
\end{align}
\end{proposition}

\begin{proof}
The extension theorem gives positive contractions \(\Sigma_s^{(p)}\).
The modular group law on \(j_p(M)\) extends by density, so
\(\Sigma_{-s}^{(p)}\) is the inverse of \(\Sigma_s^{(p)}\) and both
are isometries.  Strong continuity follows from the point strong-\(*\)
continuity of modular flow and Lemma~\ref{lem:sandwich-continuity},
first on \(j_p(M)\) and then on \(L^p(M)\).

Similarly, \eqref{eq:Tc-modular-covariance} gives
\[
    \Sigma_s^{(p)}V_c^{(p)}j_p(a)
    =j_p\!\left(T_{ce^{-2\pi s}}(\sigma_s^\omega(a))\right)
    =V_{ce^{-2\pi s}}^{(p)}\Sigma_s^{(p)}j_p(a),
\]
and density proves \eqref{eq:Ln-Borchers-covariance}.  Rewriting it as
\[
    V_c^{(p)}\Sigma_s^{(p)}
    =\Sigma_s^{(p)}V_{ce^{2\pi s}}^{(p)}
\]
and taking the difference quotient at \(c=0\) proves
\(\Sigma_s^{(p)}\Dom(G_p)\subset\Dom(G_p)\) and the generator
identity in \eqref{eq:Ln-generator-modular-covariance}.  Replacing
\(s\) by \(-s\) gives equality of the domains.
\end{proof}

%%%%%%%%%%%%%%%%%%%%%%%%%%%%%%%%%%%%%%%%%%%%%%%%%%%%%%%%%%%%%%%%%%%%%%%%%%%%%%
\subsection{Transport of sandwiched densities}
\label{subsec:Ln-density-transport}

We now use the complementary-exponent duality to identify
the \(L^p\)-orbit with the sandwiched densities of the null-translated
states.

\begin{lemma}[Transport of sandwiched densities]
\label{lem:density-propagation-fixed}
Let $1<p<\infty$, and let \(0\neq\psi\in M_*^+\) satisfy
\(Q_p(\psi\Vert\omega;M)<\infty\), and set
\[
    x_0\coloneqq x_{\psi\mid\omega,p},
    \qquad
    x_c\coloneqq V_c^{(p)}x_0,
    \qquad c\geq0.
\]
For the translated functional
\(\psi_c=\psi\circ T_c\) from
\eqref{eq:null-translated-functional}, one has
\begin{equation}
    h_{\psi_c}
    =
    h^{\eta_p}x_ch^{\eta_p},
    \qquad
    x_c=x_{\psi_c\mid\omega,p}.
    \label{eq:Ln-density-propagation}
\end{equation}
Consequently,
\begin{equation}
    0<
    Q_p(\psi_c\Vert\omega;M)
    =
    \operatorname{tr}(x_c^p)
    <\infty,
    \qquad c\geq0.
    \label{eq:Ln-moment-propagation}
\end{equation}
\end{lemma}

\begin{proof}
Because \(1/(2q)=\eta_p\), the dual symmetric embedding is
\[
    j_{q}(a)=h^{\eta_p}ah^{\eta_p}.
\]
For \(a\in M\), Proposition~\ref{prop:equal-dual-semigroups} and
\eqref{eq:main-Lp-extension} give
\begin{align}
    \operatorname{tr}\!\left(x_cj_{q}(a)\right)
    &=
    \operatorname{tr}\!\left(
      x_0V_c^{(q)}j_{q}(a)
    \right)
    \notag\\
    &=
    \operatorname{tr}\!\left(
      x_0j_{q}(T_c(a))
    \right)
    \notag\\
    &=
    \operatorname{tr}\!\left(h_\psi T_c(a)\right)
     =\psi_c(a).
    \label{eq:Ln-density-transport-pairing}
\end{align}
Here the third equality uses
\(h_\psi=h^{\eta_p}x_0h^{\eta_p}\) and the cyclicity of
\(\operatorname{tr}\) for H\"older-compatible products,
Lemma~\ref{lem:holder-products-trace-cyclicity}.
The same cyclicity gives
\[
    \operatorname{tr}\!\left(x_cj_{q}(a)\right)
    =
    \operatorname{tr}\!\left(
      h^{\eta_p}x_ch^{\eta_p}a
    \right).
\]
Thus \(h^{\eta_p}x_ch^{\eta_p}\) and \(h_{\psi_c}\) represent the same
normal functional on \(M\), proving the first identity in
\eqref{eq:Ln-density-propagation}.  The second follows from the uniqueness
in Definition~\ref{def:SRD}.  Finally,
\(\psi_c(1)=\psi(1)>0\), so \(x_c\neq0\).  The moment formula
\eqref{eq:Ln-moment-propagation} follows from
Definition~\ref{def:SRD}.
\end{proof}

The null deformation is therefore represented on the fixed space \(L^p(M)\)
by the orbit \(x_c=V_c^{(p)}x_0\) of the positive contraction semigroup.  

%%%%%%%%%%%%%%%%%%%%%%%%%%%%%%%%%%%%%%%%%%%%%%%%%%%%%%%%%%%%%%%%%%%%%%%%%%%%%%
\section{Log-convexity of \texorpdfstring{\(L^n(M)\)}{Ln(M)} norms}
\label{sec:integer-renyi-qnec}

Fix an integer \(n\geq2\), and let \(V^{(n)}\) and \(G_n\) be the
semigroup and its generator from Section~\ref{sec:Ln-dynamics}. 
We now proceed to the proof of log-convexity of $Q_n(c)=\operatorname{tr}(x_c^n)$ 
for every non-zero $x_0\in L^n(M)_+$, with $x_c=V_c^{(n)}x_0$.
Regularity for $c>0$ allows us to differentiate this moment at every positive cut.
The cyclic Ward 
identity proved next then allows us to prove 
the desired log-convexity using the Cauchy-Schwarz inequality in 
the Hilbert space $L^2(M)$. This is extendable to the initial cut 
by continuity. We subsequently characterize saturation and prove strict
decrease and strict convexity of the SRD for nonstationary orbits.
Section~\ref{subsec:rqnec-spectral} gives an independent spectral proof
of log-convexity at order two.

%%%%%%%%%%%%%%%%%%%%%%%%%%%%%%%%%%%%%%%%%%%%%%%%%%%%%%%%%%%%%%%%%%%%%%%%%%%%%%%%%%%%%%%%%%%%%%%%%%%%
\subsection{The cyclic Ward identity}
\label{sec:cyclic-ward}

Differentiating \(\operatorname{tr}(x_c^n)\) twice produces a term involving \(G_n^2x_c\). 
The cyclic Ward identity expresses this term through products containing two factors of \(G_nx_c\). 
For positive orbits, this allows us to write the second derivative as a positive combination of squared \(L^2(M)\)-norms and compare it with the first derivative by Cauchy--Schwarz. 
Since \(G_n\) does not satisfy a general Leibniz rule on products of Haagerup \(L^n\)-elements, we first work with a test core of bounded observables that are norm-differentiable under null translations and entire analytic under modular flow. On this core, the trace of an \(n\)-fold product can be written as an \(\omega\)-expectation with the observables placed at imaginary modular times \(0,-i/n,\ldots,-(n-1)i/n\). Its derivative under a common null translation vanishes by the invariance of the reference state under null translations . The Leibniz rule for the null-translation generator on bounded observables, together with Borchers covariance, then gives the cyclic identity with coefficients given by powers of \(\zeta_n\coloneqq e^{-2\pi i/n}\). Finally, graph-norm approximation extends the identity to the full domain of \(G_n\).

For \(y_1,\ldots,y_n\in L^n(M)\), define
\begin{equation}
    \Lambda_n(y_1,\ldots,y_n)
    \coloneqq
    \operatorname{tr}(y_1\cdots y_n).
    \label{eq:cyclic-n-linear-form}
\end{equation}
By the generalized H\"older inequality
of Lemma~\ref{lem:holder-products-trace-cyclicity}, the product belongs to \(L^1(M)\) and
\begin{equation}
    \bigl|\Lambda_n(y_1,\ldots,y_n)\bigr|
    \leq
    \prod_{r=1}^n\|y_r\|_n.
    \label{eq:cyclic-n-linear-bound}
\end{equation}
Thus \(\Lambda_n\) is a continuous \(n\)-linear form.

We first consider observables whose right derivative under \(T_c\)
exists in operator norm:
\begin{align}
    \Dom(\delta)
    &\coloneqq
    \left\{
      a\in M:
      \lim_{c\downarrow0}
      \frac{T_c(a)-a}{c}
      \ \text{exists in operator norm}
    \right\},
    \notag\\
    \delta a
    &\coloneqq
    \lim_{c\downarrow0}
    \frac{T_c(a)-a}{c}.
    \label{eq:null-endomorphism-generator}
\end{align}
Multiplicativity of \(T_c\) implies that \(\Dom(\delta)\) is an algebra
and that \(\delta\) obeys the Leibniz rule on its domain. Appendix~\ref{appsec:Ward-core-construction} constructs a
test core \(\mathcal{W}\subset\Dom(\delta)\) by null-semigroup smoothing 
followed by Gaussian modular smoothing.  Observables in \(\mathcal{W}\)  
are differentiable in operator norm under \(T_c\) and have entire 
analytic continuations for modular flow.

\begin{proposition}[Ward test core]
\label{prop:Ward-core-toolkit}
There is a linear subspace \(\mathcal W\subset\Dom(\delta)\)
with the following properties.
\begin{enumerate}
    \item For \(a\in\mathcal{W}\), both \(a\) and \(\delta a\) are entire analytic for
    \(\sigma^\omega\), with
    \(\sigma_z^\omega(a)\in\Dom(\delta)\) for every \(z\in\mathbb C\),
    and
    \begin{equation}
        \delta\!\left(\sigma_z^\omega(a)\right)
        =
        e^{2\pi z}\sigma_z^\omega(\delta a).
        \label{eq:complex-Borchers-on-Ward-core}
    \end{equation}

    \item For \(1\leq p<\infty\), the subspace
    \begin{equation}
        \mathcal D_p^{\mathrm W}\coloneqq j_p(\mathcal W)
        \label{eq:Ln-Ward-core}
    \end{equation}
    is a graph core for \(G_p\).  On this core,
    \begin{equation}
        G_pj_p(a)=j_p(\delta a),
        \qquad a\in\mathcal W.
        \label{eq:Gn-on-Ward-core}
    \end{equation}
\end{enumerate}
\end{proposition}

For integer $n$, using cyclicity of \(\operatorname{tr}(j_n(a_1)\cdots j_n(a_n))\) from 
Lemma~\ref{lem:holder-products-trace-cyclicity}, the factors \(h^{1/(2n)}\) combine 
into \(n\) factors \(h^{1/n}\).  For modular-entire \(a_1,\ldots,a_n\), commuting
these powers past the observables gives the following KMS placement formula proved in Appendix~\ref{appsubsec:KMS-and-integration},
\begin{equation}
\begin{aligned}
    \Lambda_n\bigl(j_n(a_1),\ldots,j_n(a_n)\bigr)
    =
    \omega\bigl(
       a_1
       \sigma_{-i/n}^\omega(a_2)
       \sigma_{-2i/n}^\omega(a_3)
       \cdots
       \sigma_{-(n-1)i/n}^\omega(a_n)
    \bigr).
\end{aligned}
    \label{eq:KMS-placement-Ward}
\end{equation}

\begin{theorem}[Cyclic Ward identity]
\label{thm:full-domain-ward-fixed}
For every \(y_1,\ldots,y_n\in\Dom(G_n)\),
\begin{equation}
    \sum_{k=0}^{n-1}\zeta_n^k
    \Lambda_n\bigl(
      y_1,\ldots,y_k,
      G_ny_{k+1},
      y_{k+2},\ldots,y_n
    \bigr)
    =0.
    \label{eq:full-domain-ward-fixed}
\end{equation}
% Empty strings of arguments at the two endpoints of the sum are omitted.
\end{theorem}

\begin{proof}
We first prove the identity on the Ward test core.  Write
\(y_r=j_n(a_r)\), where \(a_1,\ldots,a_n\in\mathcal W\), and set
\begin{equation}
    b_r\coloneqq
    \sigma_{-(r-1)i/n}^\omega(a_r),
    \qquad
    B\coloneqq b_1\cdots b_n.
    \label{eq:Ward-product-B}
\end{equation}
Proposition~\ref{prop:Ward-core-toolkit} and the derivation property of
\(\delta\) imply that \(B\in\Dom(\delta)\).  Since
\(\omega\circ T_c=\omega\),
\[
    0
    =
    \left.\frac{d}{dc}\right|_{c=0^+}
    \omega(T_c(B))
    =
    \omega(\delta B).
\]
Applying the Leibniz rule to \(\delta B\) gives
\begin{equation}
    0
    =
    \sum_{r=1}^n
    \omega\!\left(
      b_1\cdots b_{r-1}(\delta b_r)b_{r+1}\cdots b_n
    \right).
    \label{eq:Ward-core-Leibniz-sum}
\end{equation}
Equation~\eqref{eq:complex-Borchers-on-Ward-core}, evaluated at
\(z=-(r-1)i/n\), yields
\begin{equation}
    \delta b_r
    =
    \zeta_n^{\,r-1}
    \sigma_{-(r-1)i/n}^\omega(\delta a_r).
    \label{eq:Ward-root-of-unity}
\end{equation}
By the KMS placement formula \eqref{eq:KMS-placement-Ward} and
\eqref{eq:Gn-on-Ward-core}, the \(r\)-th summand in
\eqref{eq:Ward-core-Leibniz-sum} is therefore
\[
    \zeta_n^{\,r-1}
    \Lambda_n\bigl(
      y_1,\ldots,y_{r-1},
      G_ny_r,
      y_{r+1},\ldots,y_n
    \bigr).
\]
This proves \eqref{eq:full-domain-ward-fixed} on
\((\mathcal D_n^{\mathrm W})^n\).

It remains to pass to the full generator domain.  Write the graph norm 
on \(\Dom(G_n)\) as 
\[
    \|y\|_{G_n}\coloneqq\|y\|_n+\|G_ny\|_n.
\]
For each \(k\), H\"older's inequality yields
\begin{equation}
\begin{aligned}
    &\left|
      \Lambda_n\bigl(
        y_1,\ldots,y_k,
        G_ny_{k+1},
        y_{k+2},\ldots,y_n
      \bigr)
    \right|%\\ &\hspace{22mm}
    \leq
      \|G_ny_{k+1}\|_n
      \prod_{r\neq k+1}\|y_r\|_n.
\end{aligned}
    \label{eq:Ward-form-graph-bound}
\end{equation}
The left-hand side of \eqref{eq:full-domain-ward-fixed} is therefore
continuous on \(\Dom(G_n)^n\) in the product graph norm.  Since
\(\mathcal D_n^{\mathrm W}\) is a graph core by
Proposition~\ref{prop:Ward-core-toolkit}, the identity extends by density
to arbitrary \(y_1,\ldots,y_n\in\Dom(G_n)\).
\end{proof}

For the proof of RQNEC, we need the following consequence.
Because it follows from the multilinear identity, we do not need to
construct a graph core for \(G_n^2\).

\begin{corollary}[Ward identity for the \(n\)-th moment]
\label{cor:cyclic-identity-logconvexity-form}
If \(x\in\Dom(G_n^2)\), then
\begin{equation}
\begin{aligned}
    0={}
    \operatorname{tr}\!\left((G_n^2x)x^{n-1}\right)+
    \sum_{j=0}^{n-2}
    \zeta_n^{j+1}
    \operatorname{tr}\!\left(
      (G_nx)x^j(G_nx)x^{n-2-j}
    \right).
\end{aligned}
    \label{eq:cyclic-identity-logconvexity-form}
\end{equation}
\end{corollary}

\begin{proof}
Apply Theorem~\ref{thm:full-domain-ward-fixed} to
\[
    (y_1,\ldots,y_n)=(G_nx,x,\ldots,x).
\]
The \(k=0\) term is
\(\operatorname{tr}((G_n^2x)x^{n-1})\).  In the remaining terms, put
\(j=k-1\).
\end{proof}

%%%%%%%%%%%%%%%%%%%%%%%%%%%%%%%%%%%%%%%%%%%%%%%%%%%%%%%%%%%%%%%%%%%%%%%%%%%%%%
\subsection{Log-convexity of integer moments}
\label{subsec:smooth-orbit-log-convexity-fixed}

The cyclic Ward identity and interior regularity now give log-convexity
for arbitrary positive initial data in $L^n(M)$.

\begin{proof}[Proof of Theorem~\ref{thm:general-log-convexity}]
 By
Lemma~\ref{lem:positive-Ln-sandwiched-density}, \(x_0\) is the
sandwiched density of a non-zero functional \(\psi\in M_*^+\).
Lemma~\ref{lem:density-propagation-fixed} identifies \(x_c\) with
the sandwiched density of \(\psi_c=\psi\circ T_c\).
Since \(T_c\) is unital, \(\psi_c(1)=\psi(1)>0\), so \(x_c\neq0\).
Contractivity of \(V^{(n)}\) therefore gives
\begin{equation}
    0<Q_n(c)=\|x_c\|_n^n
    \leq\|x_0\|_n^n<\infty,
    \qquad c\geq0.
    \label{eq:section5-Q-positive}
\end{equation}

By Lemma~\ref{lem:null-semigroup-analytic}, the orbit $c\mapsto x_c$
is norm analytic on $(0,\infty)$ and satisfies
\[
    x_c\in\bigcap_{k\geq1}\operatorname{Dom}(G_n^k),
    \qquad
    \frac{d}{dc}x_c=G_nx_c,
    \qquad
    \frac{d^2}{dc^2}x_c=G_n^2x_c,
    \qquad c>0.
\]
The continuous $n$-linear product map and the continuity of the
Haagerup trace show that $x\mapsto\operatorname{tr}(x^n)$ is a
continuous polynomial on $L^n(M)$. Hence $Q_n$, and also $\log Q_n$,
is real analytic on $(0,\infty)$.
Fix \(c>0\) and write
\begin{equation}
    x\coloneqq x_c,
    \qquad
    u\coloneqq G_nx_c,
    \qquad
    v\coloneqq G_n^2x_c.
    \label{eq:section5-xuv}
\end{equation}
Since \(x_c=x_c^*\) for every \(c\) and the involution is continuous on
\(L^n(M)\), differentiation gives \(u=u^*\) and \(v=v^*\).

By Lemma~\ref{lem:holder-products-trace-cyclicity}, the \(n\)-fold product
map from \(L^n(M)^n\) to \(L^1(M)\) is continuous.  The finite product
rule for continuous multilinear maps and cyclicity of $\operatorname{tr}$ from Lemma~\ref{lem:holder-products-trace-cyclicity} therefore give
\begin{align}
    Q_n'(c)&=nC,
    &
    C&\coloneqq\operatorname{tr}(ux^{n-1}),
    \label{eq:section5-first-derivative}\\
    \frac{Q_n''(c)}{n}
    &=A+\sum_{j=0}^{n-2}B_j,
    &
    A&\coloneqq\operatorname{tr}(vx^{n-1}),
    \notag\\
    &&
    B_j&\coloneqq
    \operatorname{tr}\!\left(ux^jux^{n-2-j}\right),
    \qquad 0\leq j\leq n-2.
    \label{eq:section5-second-derivative}
\end{align}
Here and below a zero power of \(x\) means the identity.  All products in
these formulas belong to \(L^1(M)\).

The Ward identity removes the term containing the second generator.
Since \(x\in\Dom(G_n^2)\),
Corollary~\ref{cor:cyclic-identity-logconvexity-form} gives
\[
    A+\sum_{j=0}^{n-2}\zeta_n^{j+1}B_j=0.
\]
Substitution into \eqref{eq:section5-second-derivative} yields
\begin{equation}
    \frac{Q_n''(c)}{n}
    =
    \sum_{j=0}^{n-2}
    \bigl(1-\zeta_n^{j+1}\bigr)B_j.
    \label{eq:section5-Ward-substitution}
\end{equation}

To determine the sign of this expression, define
\begin{equation}
    Y_j\coloneqq
    x^{j/2}u x^{(n-2-j)/2}\in L^2(M),
    \qquad 0\leq j\leq n-2.
    \label{eq:section5-Yj}
\end{equation}
H\"older multiplication and cyclicity give
\begin{equation}
    B_j
    =\operatorname{tr}(Y_j^*Y_j)
    =\|Y_j\|_2^2
    \geq0.
    \label{eq:section5-Bj-positive}
\end{equation}
Moreover, \(Y_{n-2-j}=Y_j^*\), and hence
\begin{equation}
    B_{n-2-j}=B_j.
    \label{eq:section5-Bj-reflection}
\end{equation}
Under the involution \(j\mapsto n-2-j\), the phase
\(\zeta_n^{j+1}\) is replaced by its complex conjugate, while \(B_j\)
is unchanged.  Pairing these terms in
\eqref{eq:section5-Ward-substitution} therefore gives
\begin{equation}
    \frac{Q_n''(c)}{n}
    =\sum_{j=0}^{n-2}w_jB_j,
    \qquad
    w_j\coloneqq
    1-\cos\!\left(\frac{2\pi(j+1)}{n}\right)
    =2\sin^2\!\left(\frac{\pi(j+1)}{n}\right).
    \label{eq:section5-positive-Ward-sum}
\end{equation}
In particular, \(w_j\geq0\), and the sum of the \(n\)-th roots of unity
gives
\begin{equation}
    \sum_{j=0}^{n-2}w_j=n.
    \label{eq:section5-Ward-weight-sum}
\end{equation}

It remains to compare each \(B_j\) with the first derivative.  Put
\(Z\coloneqq x^{n/2}\in L^2(M)\).  Then
\begin{equation}
    \|Z\|_2^2=Q_n(c),
    \qquad
    \|Y_j\|_2^2=B_j,
    \qquad
    \langle Y_j,Z\rangle_{L^2(M)}=C,
    \label{eq:section5-L2-data}
\end{equation}
where the last equality follows from cyclicity:
\[
    \operatorname{tr}(Y_j^*Z)
    =
    \operatorname{tr}\!\left(
      x^{(n-2-j)/2}u x^{j/2}x^{n/2}
    \right)
    =
    \operatorname{tr}(ux^{n-1}).
\]
Since \(Q_n\) is real-valued,
\eqref{eq:section5-first-derivative} shows that \(C\in\mathbb R\).
The Cauchy--Schwarz inequality in \(L^2(M)\) now gives
\begin{equation}
    C^2\leq Q_n(c)B_j,
    \qquad 0\leq j\leq n-2.
    \label{eq:section5-Cauchy-Schwarz}
\end{equation}
Combining \eqref{eq:section5-positive-Ward-sum},
\eqref{eq:section5-Ward-weight-sum}, and
\eqref{eq:section5-Cauchy-Schwarz}, we obtain
\begin{align}
    Q_n(c)Q_n''(c)
    &=nQ_n(c)\sum_{j=0}^{n-2}w_jB_j
      \notag\\
    &\geq n\sum_{j=0}^{n-2}w_jC^2
     =n^2C^2
     =Q_n'(c)^2.
    \label{eq:section5-differential-log-convexity}
\end{align}
By \eqref{eq:section5-Q-positive}, this is equivalent to
\[
    (\log Q_n)''(c)\geq0,
    \qquad c>0.
\]
Thus \(\log Q_n\) is convex on \((0,\infty)\).  Finally, strong
continuity of \(V^{(n)}\) and the identity
\(Q_n(c)=\|x_c\|_n^n\) show that \(Q_n\), and hence \(\log Q_n\), is
continuous at \(c=0\).  Convexity therefore extends to
\([0,\infty)\), proving \eqref{eq:regular-orbit-log-convexity}.
\end{proof}

\begin{remark}[The convexity deficit]
For any non-zero $x_0\in L^n(M)_+$ and any $c>0$, use the notation
of the preceding proof and put $a_c=C/Q_n(c)$. Then
\begin{equation}
\label{eq:rqnec-convexity-deficit}
    (n-1)S_n''(c)=(\log Q_n)''(c)
    =\frac{n}{Q_n(c)}
      \sum_{j=0}^{n-2}w_j\|Y_j-a_cZ\|_2^2,
\end{equation}
where $S_n$ is the SRD of the corresponding positive functional.
Since every $w_j$ is strictly positive, saturation at $c$ is
equivalent to $Y_j=a_cZ$ for every $j$.
\end{remark}

%%%%%%%%%%%%%%%%%%%%%%%%%%%%%%%%%%%%%%%%%%%%%%%%%%%%%%%%%%%%%%%%%%%%%%%%%%%%%%
\subsection{Saturation and strict convexity}
\label{subsec:integer-renyi-qnec-final-fixed}

We now characterize equality and show that the moment of
every nonstationary orbit is strictly decreasing and strictly
log-convex at positive cuts.

\begin{proposition}[Saturation and strict convexity]
\label{prop:rqnec-saturation}
Let $n\geq2$ be an integer and let $0\ne x_0\in L^n(M)_+$. Set
\[
    x_c:=V_c^{(n)}x_0,
    \qquad Q_n(c):=\operatorname{tr}(x_c^n).
\]
For every $c_*>0$,
\[
    (\log Q_n)''(c_*)=0
    \quad\Longleftrightarrow\quad G_nx_{c_*}=0
    \quad\Longleftrightarrow\quad
    x_c=x_0\quad\text{for all }c\geq0.
\]
Consequently, unless the orbit is stationary,
\[
    (\log Q_n)'(c)<0,
    \qquad (\log Q_n)''(c)>0,
    \qquad c>0.
\]
\end{proposition}

\begin{proof}
By Lemma~\ref{lem:null-semigroup-analytic}, the orbit
$c\mapsto x_c$ is norm analytic on $(0,\infty)$, and
\[
    x_c\in\operatorname{Dom}(G_n^k),
    \qquad c>0,\quad k\geq1.
\]
Lemma~\ref{lem:holder-products-trace-cyclicity} shows
that $x\mapsto\operatorname{tr}(x^n)$ is a continuous
polynomial on $L^n(M)$. Hence $Q_n$ is real analytic
on $(0,\infty)$, and the same holds for $\log Q_n$, since $Q_n(c)>0$ .

Fix $c_*>0$ and put
\[
    x\coloneqq x_{c_*},
    \qquad
    u\coloneqq G_nx_{c_*},
    \qquad
    a_*\coloneqq
    \frac{\operatorname{tr}(ux^{n-1})}{Q_n(c_*)}.
\]
The orbit is positive, so $u=u^*$, and trace
cyclicity gives $a_*\in\mathbb R$.
Applying \eqref{eq:rqnec-convexity-deficit} directly at $c_*$, and
using strict positivity of the weights $w_j$, shows that
$(\log Q_n)''(c_*)=0$ precisely when
\[
    x^{j/2}ux^{(n-2-j)/2}=a_*x^{n/2},
    \qquad 0\leq j\leq n-2.
\]

Suppose these equalities hold. For $n=2$, they give
$u=a_*x$ directly. For $n>2$, their two endpoint
cases give, with $\beta=(n-2)/2$,
\[
    ux^\beta=a_*x^{\beta+1},
    \qquad
    x^\beta u=a_*x^{\beta+1}.
\]
Put $e=s(x)$ and $f=1-e$. The support $e$ belongs
to $M$, because the dual-action homogeneity of $x$
makes its support $e$ invariant under the dual action.
To cancel $x^\beta$ on this support, set
\[
    d=u-a_*x,
    \qquad
    e_k=\mathbf1_{[1/k,k]}(x).
\]
Since $x^{-\beta}e_k$ is bounded, multiplication
of $dx^\beta=0$ by this operator gives $de_k=0$.
As $e_k\uparrow e$ strongly, closedness of $d$
gives $de=0$. Taking adjoints gives $ed=0$.
Consequently,
\[
    ue=eu=a_*x,
    \qquad
    u-a_*x=fuf.
\]

For sufficiently small $t$ of either sign,
$fx_{c_*+t}f\geq0$, whereas $fx_{c_*}f=0$.
Thus
\[
    \frac{fx_{c_*+t}f}{t}
    \longrightarrow fuf
    \quad\text{in }L^n(M).
\]
These difference quotients are positive for $t>0$
and negative for $t<0$. Since the positive cone
is norm closed, their common limit is zero.
It follows that $u=a_*x$.

By Lemmas~\ref{lem:positive-Ln-sandwiched-density}
and~\ref{lem:density-propagation-fixed}, there is a
non-zero $\psi\in M_*^+$ such that
\[
    \mu\coloneqq
    \operatorname{tr}(h^{\eta_n}x_ch^{\eta_n})
    =\psi(1)>0,
    \qquad
    \eta_n=\frac{n-1}{2n},
\]
independently of $c$. H\"older's inequality makes
the sandwich map $h^{\eta_n}\cdot h^{\eta_n}$ from $L^n(M)$ to $L^1(M)$ bounded.
We may therefore differentiate this identity at
$c_*$ to obtain
\[
    0=\operatorname{tr}(h^{\eta_n}uh^{\eta_n})
     =a_*\mu.
\]
Hence $a_*=0$ and $G_nx_{c_*}=u=0$.

Whenever $G_nx_{c_*}=0$, the semigroup property
and the generator differentiation formula give
\[
    x_{c_*+t}
    =V_t^{(n)}x_{c_*}
    =x_{c_*},
    \qquad t\geq0.
\]
The norm-analytic orbit is therefore constant on
a nonempty interval. Applying the real-analytic
identity theorem to its scalar pairings with
elements of $(L^n(M))^*$ shows that
$x_c=x_{c_*}$ for every $c>0$. Strong continuity
at zero then gives $x_0=x_{c_*}$.
Conversely, a stationary orbit has
$G_nx_{c_*}=0$ and $(\log Q_n)''(c_*)=0$.
This proves the stated equivalences.

Finally, Theorem~\ref{thm:general-log-convexity}
gives $(\log Q_n)''(c)\geq0$ for every $c>0$.
If the orbit is not stationary, the equivalences
just proved exclude equality at every positive
cut. Thus $(\log Q_n)''(c)>0$ throughout
$(0,\infty)$.

Contractivity and the semigroup property imply that $Q_n$ is
non-increasing. Thus $(\log Q_n)'\leq0$ on $(0,\infty)$.
For a nonstationary orbit, this derivative is strictly increasing
because $(\log Q_n)''>0$. If it vanished at a positive cut, it would
be positive at every later cut, contradicting monotonicity. Hence
$(\log Q_n)'<0$ throughout $(0,\infty)$.
\end{proof}

\subsection{Spectral representation at order two}
\label{subsec:rqnec-spectral}

By Lemma~\ref{lem:null-l2-implementation}, the null-translation
semigroup on $L^2(M)$ is unitarily equivalent to $e^{-cP}$.
Its squared norm therefore admits a Laplace representation. This
gives an independent proof of order-two log-convexity that does not
use the cyclic Ward identity.

\begin{proposition}[Direct log-convexity at order two]
\label{prop:n2-direct-log-convexity}
Let \(0\neq x_0\in L^2(M)_+\), and define \(x_c\) and \(Q_2(c)\) by
\eqref{eq:main-Qn} with \(n=2\).  If \(\xi_0\coloneqq Wx_0\), then
\begin{equation}
    Q_2(c)
    =\left\langle
       \xi_0,e^{-2cP}\xi_0
     \right\rangle_{\mathcal H},
    \qquad c\geq0.
    \label{eq:n2-spectral-moment}
\end{equation}
Consequently, \(Q_2\) is strictly positive, finite, continuous,
non-increasing, and log-convex on \([0,\infty)\).
\end{proposition}

\begin{proof}
Lemma~\ref{lem:null-l2-implementation} gives
\(Wx_c=e^{-cP}\xi_0\).  Since \(x_c\geq0\), \eqref{eq:n2-spectral-moment} follows from
\begin{equation}
    Q_2(c)
    =\operatorname{tr}(x_c^2)
    =\|x_c\|_2^2
    =\|e^{-cP}\xi_0\|_{\mathcal H}^2.
    \label{eq:n2-moment-norm}
\end{equation}

Let \(E_P\) be the spectral measure of \(P\), and define the finite
non-zero positive measure
\begin{equation}
    \mu(B)\coloneqq\|E_P(B)\xi_0\|_{\mathcal H}^2,
    \qquad B\subset[0,\infty)\ \text{Borel}.
    \label{eq:n2-spectral-measure}
\end{equation}
The positivity of \(P\) and the spectral theorem give
\begin{equation}
    Q_2(c)
    =\int_{[0,\infty)}e^{-2c\lambda}\,d\mu(\lambda).
    \label{eq:n2-Laplace-transform}
\end{equation}
Strict positivity, finiteness, continuity, and monotonicity are immediate
from this representation.  Moreover, for \(c_0,c_1\geq0\) and
\(0<\theta<1\), the scalar H\"older inequality yields
\begin{align}
    Q_2\bigl((1-\theta)c_0+\theta c_1\bigr)
    &=\int
      \bigl(e^{-2c_0\lambda}\bigr)^{1-\theta}
      \bigl(e^{-2c_1\lambda}\bigr)^\theta
      \,d\mu(\lambda)
      \notag\\
    &\leq
      Q_2(c_0)^{1-\theta}Q_2(c_1)^\theta.
    \label{eq:n2-global-log-convexity}
\end{align}
The inequality also holds for \(\theta=0,1\).  Hence \(Q_2\) is log-convex
on \([0,\infty)\) for all \(x_0\in L^2(M)_+\).
\end{proof}

Now let \(0\neq\psi\in M_*^+\) satisfy
\(D_2(\psi\Vert\omega;M)<\infty\), and take
\(x_0=x_{\psi\mid\omega,2}\).  Equation
\eqref{eq:main-four-level-dictionary} becomes
\begin{equation}
    D_2(\psi\Vert\omega;M_c)
    =\log Q_2(c)-\log\psi(1).
    \label{eq:n2-SRD-spectral-form}
\end{equation}
Since the second term in \eqref{eq:n2-SRD-spectral-form} is independent of
\(c\), Proposition~\ref{prop:n2-direct-log-convexity} proves the \(n=2\)
case of Corollary~\ref{cor:integer-renyi-qnec-fixed} without the cyclic Ward
identity.  If the HSMI is realized by geometric null-cut algebras and
\eqref{eq:weighted-ane-generator} holds, then
\eqref{eq:n2-spectral-moment} represents the order-two R\'{e}nyi moment as
a Euclidean null-translation matrix element of the weighted averaged
null-energy operator.

\begin{remark}[Energy and fluctuations at order two]
\label{rem:n2-energy-variance}
For $c>0$, define the normalized natural-cone vector
\[
    \eta_c\coloneqq\frac{Wx_c}{\sqrt{Q_2(c)}}
    =\frac{e^{-cP}Wx_0}{\|e^{-cP}Wx_0\|}.
\]
It represents an escort state \cite{Kibe:2026dtz} with Haagerup density
$h_{\varphi_c}=x_c^2/Q_2(c)$.
Since $P\geq0$, spectral calculus gives
$\eta_c\in\bigcap_{k\geq1}\operatorname{Dom}(P^k)$.
Differentiating \eqref{eq:n2-spectral-moment} and using
\eqref{eq:n2-SRD-spectral-form} yields
\begin{align}
    -S_2'(c)
    &=2\langle\eta_c,P\eta_c\rangle,
    \label{eq:n2-escort-energy}\\
    S_2''(c)
    &=4\operatorname{Var}_{\eta_c}(P),
    \label{eq:n2-escort-variance}
\end{align}
where
\[
    \operatorname{Var}_{\eta_c}(P)
    \coloneqq
    \|P\eta_c\|^2-\langle\eta_c,P\eta_c\rangle^2.
\]
Thus the rate of decrease of the order-two SRD is twice
the canonical translation energy of the escort vector,
while its curvature is four times the corresponding
energy variance. When \eqref{eq:weighted-ane-generator}
holds, these quantities are the expectation and variance
of the weighted averaged null-energy operator.
\end{remark}

\section{A sandwiched R\'enyi ant formula}
\label{sec:ant-formula}

Using information-theoretic arguments, Wall conjectured the ant formula, which gives a variational relation between the relative entropy and null energy  \cite{Wall:2017blw}.
Ceyhan and Faulkner \cite{Ceyhan:2018zfg}, and later
Hollands and Longo \cite{Hollands:2025glm}, proved it for
the canonical translations of a half-sided modular inclusion.
Here we establish an integer-order R\'enyi ant formula.
At positive cuts, finiteness of the initial sandwiched
R\'enyi divergence suffices for every even order, whereas at odd
orders, we impose the additional square-root domain condition
stated in Theorem~\ref{thm:integer-escort-ant}.
We conjecture the corresponding formula for every real
order $\alpha>1$. When the translation generator is the
weighted averaged null-energy operator, the formula gives
a variational relation between R\'enyi distinguishability
and null energy.
We verify the general-order conjecture for a coherent-state
family in the vacuum $U(1)$-current algebra on a half-line,
corresponding to one null fibre of the free scalar.

\subsection{Energy minimization and the integer-order ant formula}
\label{subsec:n2-ant-formula}

The ant formula for relative entropy relates its negative null
derivative to the infimum of averaged null energy over purifications
of the state on the cut algebra.  We prove an integer-order analogue for
the SRD, in which the state held fixed is the escort state
\eqref{eq:integer-endpoint-escort}.  We first establish the energy
minimization property of natural-cone vectors needed for this formula.

Write \(\mathcal P_\Omega=W(L^2(M)_+)\) for the natural positive cone.
For \(\zeta\in\mathcal H\), set
\begin{equation}
    q_P(\zeta)\coloneqq
    \begin{cases}
      \|P^{1/2}\zeta\|^2,
        &\zeta\in\Dom(P^{1/2}),\\
      +\infty,&\text{otherwise}.
    \end{cases}
    \label{eq:n2-energy-form}
\end{equation}
For a normal state \(\varphi\) on \(M\), let
\begin{equation}
    \mathcal C(\varphi)\coloneqq
    \left\{\zeta\in\mathcal H:
       \langle\zeta,a\zeta\rangle=\varphi(a)
       \text{ for every }a\in M
    \right\}.
    \label{eq:n2-purification-class}
\end{equation}
Every vector in \(\mathcal C(\varphi)\) is normalized.

In the following we use Hiai's convention for Connes cocycles
\cite[Section~10.2]{Hiai:2021cks}.  In terms of Haagerup densities,
\[
    (D\omega:D\varphi)_s=h_\omega^{is}h_\varphi^{-is},
\]
with imaginary powers extended by zero on the kernel when
\(\varphi\) is nonfaithful. The proof of the next lemma uses results from Ref.~\cite{Ceyhan:2018zfg}, which uses the same convention for Connes cocycles, although their relative modular operators carry the state labels
in the opposite order \cite[footnote~3]{Ceyhan:2018zfg}.

\begin{lemma}[Energy minimization for a natural-cone vector]
\label{lem:n2-cone-energy-minimization}
Let $\eta\in\mathcal P_\Omega$ be a unit vector, put
$E:=q_P(\eta)\in[0,+\infty]$, and let $\varphi$ be its state on $M$.
Denote the states induced by $\Omega$ and $\eta$ on $M'$ by
$\omega'$ and $\varphi'$, respectively. Then
\begin{equation}\label{eq:n2-cone-energy-infimum}
    \inf_{\zeta\in\mathcal C(\varphi)}q_P(\zeta)=\frac{E}{2}.
\end{equation}
The complementary cocycle vectors
\begin{equation*}
    \eta_s:=(D\omega':D\varphi')_s\eta,
    \qquad s\in\mathbb R,
\end{equation*}
belong to $\mathcal C(\varphi)$ and satisfy
\begin{equation}\label{eq:n2-cone-cocycle-energy}
    q_P(\eta_s)=\frac{E}{2}\bigl(1+e^{-2\pi s}\bigr).
\end{equation}
These identities are understood in the extended sense when
$E=+\infty$.
\end{lemma}

\begin{proof}
Put $J=J_{\Omega;M}$, $\Delta=\Delta_{\Omega;M}$,
$u_s=(D\omega:D\varphi)_s$, and
$u_s'=(D\omega':D\varphi')_s$.
Cocycles of nonfaithful states are understood with their support
projections. In particular, $(u_s')^*u_s'=s(\varphi')$, so
$\eta_s=u_s'\eta$ induces $\varphi$ on $M$ by
\cite[Eqs.~(25)--(27)]{Ceyhan:2018zfg}. Thus
$\eta_s\in\mathcal C(\varphi)$.

Suppose first that $E=+\infty$. If some
$\zeta\in\mathcal C(\varphi)$ had finite energy, its natural-cone
representative would be $\eta$, and \cite[Lemma~3]{Ceyhan:2018zfg} would give
\[
    q_P(\eta)\leq2q_P(\zeta)<\infty,
\]
a contradiction. Every vector in $\mathcal C(\varphi)$ therefore
has infinite energy. This includes every $\eta_s$, and proves
\eqref{eq:n2-cone-energy-infimum} and
\eqref{eq:n2-cone-cocycle-energy} in this case.

Now suppose that $E<\infty$. By \cite[Lemma~2]{Ceyhan:2018zfg}, there is
$R\in[0,E]$ such that
\begin{equation}
\label{eq:cocycle-affine-energy}
    E_s:=q_P(\eta_s)=R+(E-R)e^{-2\pi s}.
\end{equation}
Since $J\Omega=\Omega$ and $J\eta=\eta$, conjugation of the relative
modular operators gives $Ju_s'J=u_{-s}$.
The identity $\eta_s=\Delta^{-is}u_s\eta$ from
\cite[Eq.~(30)]{Ceyhan:2018zfg} therefore implies
\[
    J\eta_s=u_{-s}\eta=\Delta^{-is}\eta_{-s}.
\]
Borchers covariance and $JPJ=P$ give
\[
    E_s=e^{-2\pi s}E_{-s}.
\]
Comparing this with \eqref{eq:cocycle-affine-energy} yields $R=E/2$,
proving \eqref{eq:n2-cone-cocycle-energy}.

If $\zeta\in\mathcal C(\varphi)$ has finite energy, its natural-cone
representative is $\eta$, and \cite[Lemma~3]{Ceyhan:2018zfg} gives
$E\leq2q_P(\zeta)$. This inequality also holds when
$q_P(\zeta)=+\infty$. The lower bound $E/2$ is approached by
$\eta_s$ as $s\to+\infty$, which proves
\eqref{eq:n2-cone-energy-infimum}.
\end{proof}

For an integer $n\geq2$, put
\[
    x_{n,c}:=x_{\psi_c\mid\omega,n},
    \qquad
    Q_n(c):=\operatorname{tr}(x_{n,c}^n),
\]
and define the escort state by \cite{Kibe:2026dtz}
\begin{equation}
    h_{\varphi_{n,c}}
    =
    \frac{x_{n,c}^n}{Q_n(c)},
    \qquad
    \eta_{n,c}
    =
    \frac{Wx_{n,c}^{n/2}}{\sqrt{Q_n(c)}},
    \label{eq:integer-endpoint-escort}
\end{equation}
where the vector $\eta_{n,c}$ is the natural-cone representative
of $\varphi_{n,c}$.

\begin{theorem}[Integer-order R\'enyi ant formula]
\label{thm:integer-escort-ant}
Let $n\ge2$ be an integer and let
$0\ne\psi\in M_*^+$ satisfy
$D_n(\psi\Vert\omega;M)<\infty$.
If $n$ is odd, assume in addition that
\[
    x_{n,c}^{1/2}\in\operatorname{Dom}(G_{2n}).
\]
Put
\[
    \kappa_n=\frac{2n\sin(\pi/n)}{n-1}.
\]
Then, for $c>0,$
\begin{equation}
    -\frac{1}{\kappa_n}
    \frac{d}{dc}D_n(\psi\Vert\omega;M_c)
    =
    \inf_{\zeta\in\mathcal C(\varphi_{n,c})}q_P(\zeta)
    =
    \inf_{s\in\mathbb R}q_P(\eta_{n,c,s}),
    \label{eq:integer-escort-ant}
\end{equation}
where
\[
    \eta_{n,c,s}
    =
    (D\omega':D\varphi_{n,c}')_s\eta_{n,c},
    \qquad
    \varphi_{n,c}'(a')
    =
    \langle\eta_{n,c},a'\eta_{n,c}\rangle.
\]
More precisely,
\begin{equation}
    q_P(\eta_{n,c,s})
    =
    -\frac{S_n'(c)}{\kappa_n}
    \bigl(1+e^{-2\pi s}\bigr).
    \label{eq:integer-escort-cocycle-energy}
\end{equation}
% For even $n$, these conclusions hold at every $c>0$
% without an additional domain assumption.
\end{theorem}

\begin{proof}
Lemma~\ref{lem:null-semigroup-analytic} gives
$x_{n,c}\in\operatorname{Dom}(G_n)$.
Write
\[
    x=x_{n,c},
    \qquad
    u=G_nx,
    \qquad
    Q=Q_n(c),
    \qquad
    C=\operatorname{tr}(ux^{n-1}).
\]
Differentiating the $n$-fold product and using cyclicity gives
\[
    Q_n'(c)=nC,
    \qquad
    S_n'(c)=\frac{nC}{(n-1)Q}.
\]
From Lemma~\ref{lem:ward-half-power-domain},
\[
    x^{n/2}\in\operatorname{Dom}(G_2),
    \qquad
    \operatorname{tr}
    \bigl(x^{n/2}G_2(x^{n/2})\bigr)
    =
    \frac{C}{\sin(\pi/n)} \ .
\]
Since $WG_2=-PW$, it follows that
$\eta_{n,c}\in\operatorname{Dom}(P)$ and
\[
    q_P(\eta_{n,c})
    =
    -\frac{C}{Q\sin(\pi/n)}
    =
    -\frac{2S_n'(c)}{\kappa_n}.
\]
Apply Lemma~\ref{lem:n2-cone-energy-minimization}
to obtain \eqref{eq:integer-escort-ant} and
\eqref{eq:integer-escort-cocycle-energy}.
\end{proof}

% The variational constraint can also be imposed on $M_c$.
% Put
% \[
%     \widetilde\varphi_{n,c}
%     =
%     \varphi_{n,c}\circ T_c^{-1}.
% \]
% The map $\zeta\mapsto U_{-c}\zeta$ is a bijection between
% $\mathcal C(\varphi_{n,c})$ and the vector purifications of
% $\widetilde\varphi_{n,c}$ on $M_c$.
% Since $U_{-c}$ commutes with $P$, this map preserves $q_P$.

The same argument applies at $c=0$, with right derivatives, if
$x_{n,0}\in\operatorname{Dom}(G_n)$ and, for odd $n$, the square-root
hypothesis holds there. At $n=2$, the generator-domain assumption can
be removed using the Laplace representation in
Section~\ref{subsec:rqnec-spectral}. It gives
\[
    -S'_{2,+}(0)=2q_P(\eta_{2,0})
\]
as an identity in $[0,+\infty]$.
Lemma~\ref{lem:n2-cone-energy-minimization}, including its infinite-energy
case, therefore yields
\[
    -\frac14 S'_{2,+}(0)
    =\inf_{\zeta\in\mathcal C(\varphi_{2,0})}q_P(\zeta)
    =\inf_{s\in\mathbb R}q_P(\eta_{2,0,s}),
\]
since $\kappa_2=4$.

The R\'enyi ant formula 
of Theorem~\ref{thm:integer-escort-ant} naturally motivates an 
extension to all real \(\alpha>1\).
\begin{conjecture}[R\'enyi ant conjecture]
For \(\alpha>1\), suppose that
\(D_\alpha(\psi\Vert\omega;M)<\infty\), and define
\[
    x_{\alpha,c}\coloneqq x_{\psi_c\mid\omega,\alpha},
    \qquad
    h_{\varphi_{\alpha,c}}
    \coloneqq
    \frac{x_{\alpha,c}^{\alpha}}
         {\operatorname{tr}(x_{\alpha,c}^{\alpha})}.
\]
At cuts where the derivative and the energy of
the natural-cone representative are finite, we conjecture that for all \(\alpha>1\),
\begin{equation}
    -\frac{d}{dc}D_\alpha(\psi\Vert\omega;M_c)
    =
    \frac{2\alpha\sin(\pi/\alpha)}{\alpha-1}
    \inf_{\zeta\in\mathcal C(\varphi_{\alpha,c})}q_P(\zeta) \, .
    \label{eq:general-order-ant-candidate}
\end{equation}
\end{conjecture}

The usual ant formula immediately yields the QNEC: for
$c>c_0>0$, the inclusion $M_c\subset M_{c_0}$ implies that
the admissible vector purifications at cut $c$ form a superset
of those at cut $c_0$. The corresponding energy infimum
therefore cannot increase.
However, the RQNEC cannot be analogously proven from the R\'enyi ant formula, since $\mathcal C(\varphi_{\alpha,c})$ need not be nested as $c$ increases. This obstruction has been explicitly identified in terms of a transport term in \cite{Kibe:2026idb}. Heuristically, it originates from the fact that the formation of the escort state need not commute with restriction to a subregion.

\subsection{A coherent-state check at general order}
\label{subsec:coherent-general-order-check}

We verify \eqref{eq:general-order-ant-candidate} for the
coherent-state family studied in
\cite[Section~5.3]{Kibe:2026idb}. The underlying model is the chiral $U(1)$ current $\mathsf J(u)$, realized as the derivative of a free scalar
field along a null fibre. 

Let $\phi(u,0,\vec y)$ denote the free massive scalar field on the null plane $x^-=0$, with $x^+=u$. Choose a real transverse smearing $\mathcal F(\vec y)$, such that $\int \mathcal F(\vec y)^2\, d^{d-2}y=1$, and define
\begin{equation}
    \mathsf J_{\mathcal F}(u)\coloneqq \sqrt{2}\int \mathcal F(\vec y)\partial_u\phi(u,0,\vec y)\, d^{d-2}y,
\end{equation}
where we use the $\sqrt2$ normalization of \cite{Kibe:2026idb}. Henceforth we suppress the subscript $\mathcal F$. For a real test function \(f\in C_c^\infty(\mathbb R)\), define the smeared current and the Weyl operator respectively as
\begin{equation}
  \mathsf J(f)=\int_{\mathbb R}f(u) \mathsf J(u)\, du,
  \qquad \mathsf W(f)=e^{i\mathsf J(f)},
  \label{eq:current-Weyl}
\end{equation}
where the smeared current is self-adjoint after closure and $\mathsf W(f)$ is a bounded unitary. The local von Neumann algebra of interest is the half-line algebra of these Weyl operators:
\begin{equation}
     M_c=\mathcal A((c,\infty))\coloneqq \{\mathsf W(f):f\in C_c^\infty((c,\infty),\mathbb R)\}'',
\end{equation}
which corresponds to observables to the future of the
cut $u=c$. In the above normalization the specialization of
\eqref{eq:weighted-ane-generator} is
\[
    P=\int_{\mathbb R}T_{++}(u)\,du,
    \qquad
    T_{++}(u)=\frac12:\!\mathsf J(u)^2\!:,
\]
where $:\cdot:$ denotes vacuum normal ordering.

Coherent states are Gaussian excitations obtained
by applying Weyl displacements to the vacuum: $\psi_f(\cdot )\coloneqq \omega(\mathsf W(f)^* \cdot \mathsf W(f))$.
These displacements change the mean current while preserving its
vacuum connected two-point function, so the excitation is
entirely specified by its current profile $\mathfrak j_\psi(u)\coloneqq \psi(\mathsf J)$. 

Let $g>0$ and $q\in\mathbb R$.
Choose a smooth, zero-mode-free
global coherent profile which agrees with
\begin{equation}
    \mathfrak j_\psi(u)=\frac{qg}{(u+g)^2},
    \qquad u>-g/2.
    \label{eq:coherent-original-profile}
\end{equation}
A smooth completion and a compensating pulse to the past
give such a profile without changing its restriction to
any $M_c$, $c\geq0$
\cite[Eq.~(5.67)]{Kibe:2026idb}.
Let $\psi$ denote the resulting normalized state on $M$.

Put $g_c=g+c$. We use the escorts
$\varphi_{\alpha,c}$ introduced above, extending
\eqref{eq:integer-endpoint-escort} to real orders.
After pullback from $M_c$ to $M$ by $T_c$, their current
profiles are
\begin{equation}
    \mathfrak j_{\alpha,c}(u)
    =
    \frac{qg\alpha^2}{g_c^2}
    \frac{(u/g_c)^{\alpha-1}}
         {[1+(u/g_c)^\alpha]^2},
    \qquad u>0;
    \label{eq:coherent-endpoint-profile}
\end{equation}
see \cite[Eq.~(5.74)]{Kibe:2026idb}.
Thus $u$ is the distance from the physical cut, and
\eqref{eq:coherent-endpoint-profile} describes
$\varphi_{\alpha,c}$ on the fixed algebra $M$.

The direct calculation in \cite[Eqs.~(5.78)]{Kibe:2026idb} using the integral representation for the SRD \cite[Corollary~4.3]{Kibe:2026dtz} gives
\begin{equation}
     S_\alpha(c)=
    \frac{\pi q^2g^2}{12g_c^2}\alpha(\alpha+1)<\infty,
\end{equation}
for every $\alpha>1$ and $c\geq0$, 

\begin{proposition}[Coherent-state ant formula]
\label{prop:coherent-general-order-ant}
For every real $\alpha>1$ and every
$c\geq0$, the states above satisfy
\begin{equation}
    \inf_{\zeta\in\mathcal C(\varphi_{\alpha,c})}q_P(\zeta)
    =
    \frac{\pi q^2g^2(\alpha^2-1)}
         {12g_c^3\sin(\pi/\alpha)}.
    \label{eq:coherent-escort-energy-infimum}
\end{equation}
Consequently, \eqref{eq:general-order-ant-candidate}
holds for this family at every real $\alpha>1$.
\end{proposition}

\begin{proof}
We first identify the natural-cone representative of the escort. Extend its current profile to the full
line by
\[
    k_{\alpha,c}(u)
    =
    \operatorname{sgn}(u)\,
    \mathfrak j_{\alpha,c}(|u|).
\]
For $\alpha>1$, this profile belongs to
$L^1(\mathbb R)\cap L^2(\mathbb R)$, has zero integral,
and satisfies
\[
    \int_{\mathbb R}|u|\,|k_{\alpha,c}(u)|\,du<\infty.
\]
Hence the Fourier transform $\widehat k_{\alpha,c}(\nu)=O(\nu)$ at zero, and
\[
    \int_0^\infty
    \frac{|\widehat k_{\alpha,c}(\nu)|^2}{\nu}\,d\nu<\infty.
\]
The contribution from $\nu\geq1$ is finite by Plancherel.
Thus $k_{\alpha,c}$ defines a vacuum Weyl displacement
$\mathsf V_{\alpha,c}$, with current mean $k_{\alpha,c}$
and finite translation energy. The vector $\mathsf V_{\alpha,c}\Omega$
induces $\varphi_{\alpha,c}$ on $M$.

Put $J=J_{\Omega;M}$. The neutral-scalar modular reflection
\cite{Bisognano:1975ih} gives
\[
    J\mathsf J(u)J=-\mathsf J(-u).
\]
Antiunitarity and the oddness of $k_{\alpha,c}$ imply
$J\mathsf V_{\alpha,c}J=\mathsf V_{\alpha,c}$.
The Weyl relations also show that $\mathsf V_{\alpha,c}$
normalizes $M$.

Let $\mathsf U_{\alpha,c}$ be the canonical standard
implementer of the same automorphism
\cite[Chapter~3]{Hiai:2021cks}. It preserves the natural
cone and commutes with $J$. Consequently,
$\mathsf U_{\alpha,c}^*\mathsf V_{\alpha,c}$ commutes
with both $M$ and $M'$, and is scalar because $M$ is
a factor. The positive Gaussian vacuum overlap of
$\mathsf V_{\alpha,c}$, together with
$\langle\Omega,\mathsf U_{\alpha,c}\Omega\rangle\geq0$,
fixes this scalar to $1$. Therefore
\[
    \eta_{\alpha,c}
    =
    \mathsf V_{\alpha,c}\Omega
    =
    \frac{W x_{\alpha,c}^{\alpha/2}}
         {\sqrt{\operatorname{tr}(x_{\alpha,c}^{\alpha})}}
\]
is the natural-cone representative of
$\varphi_{\alpha,c}$.

Its translation energy is
\[
    q_P(\eta_{\alpha,c})
    =
    \frac12\int_{\mathbb R}k_{\alpha,c}(u)^2\,du
    =
    \int_0^\infty\mathfrak j_{\alpha,c}(u)^2\,du.
\]
Lemma~\ref{lem:n2-cone-energy-minimization} thus gives
\[
    I_\alpha(c)
    \coloneqq
    \inf_{\zeta\in\mathcal C(\varphi_{\alpha,c})}q_P(\zeta)
    =
    \frac12\int_0^\infty
    \mathfrak j_{\alpha,c}(u)^2\,du.
\]
With $t=u/g_c$, this becomes
\[
    I_\alpha(c)
    =
    \frac{q^2g^2\alpha^4}{2g_c^3}
    \int_0^\infty
    \frac{t^{2\alpha-2}}{(1+t^\alpha)^4}\,dt.
\]
The beta integral and the gamma reflection identity give
\[
\begin{aligned}
    \int_0^\infty
    \frac{t^{2\alpha-2}}{(1+t^\alpha)^4}\,dt
    &=
    \frac{1}{6\alpha}
    \Gamma(2-1/\alpha)\Gamma(2+1/\alpha)\\
    &=
    \frac{\pi(\alpha^2-1)}
         {6\alpha^4\sin(\pi/\alpha)}.
\end{aligned}
\]
This proves \eqref{eq:coherent-escort-energy-infimum}.
Finally,
\[
    \frac{2\alpha\sin(\pi/\alpha)}{\alpha-1}
    I_\alpha(c)
    =
    \frac{\pi q^2g^2\alpha(\alpha+1)}{6g_c^3}
    =
    -S_\alpha'(c),
\]
as claimed.
\end{proof}

This calculation verifies the coefficient in
\eqref{eq:general-order-ant-candidate} beyond integer
orders for this coherent-state family. Its proof uses
the explicit escort profile and the natural-cone energy
minimization lemma, independently of the integer-power
argument in Theorem~\ref{thm:integer-escort-ant}.

%%%%%%%%%%%%%%%%%%%%%%%%%%%%%%%%%%%%%%%%%%%%%%%%%%%%%%%%%%%%%%%%%%%%%%%%%%%%%%%%%

\section{Discussion}
\label{sec:discussion}

We have proved the R\'{e}nyi QNEC at every integer order
\(n\geq2\), whenever the geometric null-cut algebras form the canonical
translation family of a half-sided modular inclusion, assuming only finiteness 
of the initial sandwiched R\'{e}nyi divergence and no generator-domain 
condition on the sandwiched density. The underlying
log-convexity theorem requires no stress-tensor representation or
model-specific input.  %Finiteness of the initial sandwiched R\'{e}nyi
%divergence suffices, with   
%Comparing the resulting state domain
% with those of the ordinary-QNEC proofs in
% \cite{Ceyhan:2018zfg,Hollands:2025glm} requires a separate analysis.
The cyclic Ward identity used in the proof exploits the \(n\)-fold moment and the
\(n\)-th roots of unity to reduce log-convexity to an \(L^2(M)\)
estimate.  Non-integer powers admit no finite cyclic expansion of
this kind, so extension to arbitrary real \(\alpha>1\) requires
a new argument.
At order two, Section~\ref{subsec:rqnec-spectral} gives a direct
spectral proof by identifying the \(L^2(M)\)-semigroup with
\(e^{-cP}\).  The moment \(Q_2\) is then a Laplace transform of a
positive measure, and H\"older's inequality gives log-convexity.
For \(n>2\), the moment is no longer the squared norm of this
Hilbert-space orbit. 

Theorem~\ref{thm:integer-escort-ant} supplies a
variational energy interpretation of the first null
derivative of the SRD at integer order:
$-S_n'(c)/\kappa_n$ is the infimum of the translation
energy over vector purifications of the escort state,
approached by complementary Connes cocycle flow.
At positive cuts, finite initial divergence suffices
for even orders, while odd orders require an additional
condition on the square root of the sandwiched density.
When \eqref{eq:weighted-ane-generator} holds, the energy
being minimized is the weighted averaged null energy.
The extension \eqref{eq:general-order-ant-candidate}
to arbitrary real $\alpha>1$ remains conjectural for
general states. Subsection~\ref{subsec:coherent-general-order-check}
verifies it for the coherent-state family considered
there.

One can ask whether other state-preserving semigroups admit
covariance and Ward identities strong enough to yield moment
log-convexity, complementing the differential estimates studied in
noncommutative functional inequalities
\cite{olkiewicz1999hypercontractivity,junge2015noncommutative,
carlen2017gradient,Wirth:2025goz}. It is also natural to ask 
whether there are null-convexity bounds for optimized quantum
\(f\)-divergences \cite{Wilde:2017okz,Gao:2020rof} or the
\(\alpha\)--\(z\) R\'{e}nyi family
\cite{Audenaert:2015npv,Kato:2023aro,Kato:2023hlj,Hiai:2024qve}.

A gravitational extension could connect RQNEC with quantum focusing
\cite{Bousso:2015mna}.  Recent work in gravity linearized about
Killing-horizon backgrounds uses type \(\mathrm{II}_\infty\) cut
algebras \cite{Chandrasekaran:2026pnc} and a crossed-product description
of generalized entropy \cite{Witten:2021unn,Chandrasekaran:2022eqq}.
Our theorem cannot apply directly to proper inclusions of these
algebras, since a nontrivial HSMI of factors with a common cyclic and
separating vector must be of type \(\mathrm{III}_1\)
\cite[Theorem~2.2(4)]{Araki:2005we}.  Possible routes are to work with
the underlying type \(\mathrm{III}\) algebras or develop an analogue
for semifinite reference weights.
Such an extension would establish focusing only along the chosen
deformation ray, since universal off-diagonal sandwiched-R\'{e}nyi bounds are
already obstructed in nongravitational
theories \cite{Kibe:2026offdiag}.  Interpreting the negative derivative
of the gravitational divergence as a R\'{e}nyi quantum expansion would
also require a connection to generalized entropy and the appropriate
matter limit.

\section*{Acknowledgments}
We thank Krishna Jalan for comments on the manuscript. T.K. thanks Lautaro Saba and Gonzalo Torroba for helpful discussions. 
\section*{Statements and Declarations}

\subsection*{Funding}
T.K. is supported by a Simons Foundation fellowship through the Targeted Grant to Instituto Balseiro. The work of P.R. has been supported by the Polish National Science Centre through Sonata grant (2022/47/D/ST2/02058).

\subsection*{Competing interests}
The authors have no relevant financial or non-financial interests to disclose.

\subsection*{Data availability}
No datasets were generated or analysed during the current study.

\subsection*{Statement on AI use}
During the preparation of this manuscript, the authors used LLMs to assist with editing and the refinement of technical arguments. 
The authors developed the proof strategy for regular initial data without AI assistance, including the use of the contractive \(L^n\) semigroup and its generator, the reformulation of the sandwiched R\'enyi divergence in terms of \(L^2\) norms, and the Cauchy–Schwarz argument. 
%LLMs subsequently assisted with technical refinements of this proof. 
An LLM proposed the analyticity argument, which the authors subsequently refined and independently verified. The authors developed the R\'enyi ant theorem proof strategy, and refined it with LLMs. The authors independently checked every mathematical statement, proof, and citation in the final manuscript, and take full responsibility for its content. This work was supported in part by a grant of access to PR to OpenAI models through the ChatGPT for Academic Researchers program.

\begin{appendices}
\section{Haagerup \texorpdfstring{$L^p$}{Lp} facts}
\label{appsec:Haagerup-analytic-facts}

This appendix collects the Haagerup \(L^p\) identities used in the proof,
together with our conventions for measurable-operator products, trace
cyclicity, and symmetric embeddings.  Throughout,
\(M\subset B(\mathcal H)\) is in standard form, \(\Omega\) is cyclic and
separating, \(\omega(a)=\langle\Omega,a\Omega\rangle\),
\(h=h_\omega\), and \(j_p\) is the symmetric embedding
\eqref{eq:main-symmetric-embedding}.  Products below are closed products
of $\tau$-measurable operators affiliated with the continuous core, and
\(1/\infty=0\).

\begin{lemma}[H\"older products and trace cyclicity]
\label{lem:holder-products-trace-cyclicity}
Let \(1\leq p_1,\ldots,p_k\leq\infty\), and suppose that
\[
    \frac1r=\frac1{p_1}+\cdots+\frac1{p_k}\leq1.
\]
If \(x_j\in L^{p_j}(M)\), then
\begin{equation}
    x_1\cdots x_k\in L^r(M),
    \qquad
    \|x_1\cdots x_k\|_r
    \leq
    \prod_{j=1}^k\|x_j\|_{p_j}.
    \label{eq:appendix-general-holder-product}
\end{equation}
If \(\sum_jp_j^{-1}=1\), then
\begin{equation}
    \operatorname{tr}(x_1x_2\cdots x_k)
    =
    \operatorname{tr}(x_2\cdots x_kx_1).
    \label{eq:appendix-trace-cyclicity}
\end{equation}
\end{lemma}
\begin{proof}
    These are the generalized H\"older inequality and the trace-cyclicity
identity of \cite[Propositions~9.17 and~9.22]{Hiai:2021cks}.
\end{proof}

\begin{lemma}[Continuity of the symmetric embedding]
\label{lem:sandwich-continuity}
Let \(1\leq p<\infty\), let \(h\in L^1(M)_+\), and let
\((a_i)\) be a uniformly bounded net in \(M\) converging strongly to
\(a\).  Then
\begin{equation}
    h^{1/(2p)}a_i h^{1/(2p)}
    \longrightarrow
    h^{1/(2p)}a h^{1/(2p)}
    \qquad\text{in }L^p(M).
    \label{eq:appendix-symmetric-embedding-continuity}
\end{equation}
\end{lemma}

\begin{proof}
By \cite[Lemma~2.3]{Junge:2002osf},
\((a_i-a)h^{1/(2p)}\to0\) in \(L^{2p}(M)\).  Left multiplication by
\(h^{1/(2p)}\), followed by H\"older's inequality, gives
\eqref{eq:appendix-symmetric-embedding-continuity}.
\end{proof}

\begin{lemma}[Pairing of complementary symmetric cores]
\label{lem:complementary-core-pairing}
Let $1\leq p<\infty$, let $q\in(1,\infty]$ satisfy
$1/p+1/q=1$, and let $J=J_{\Omega;M}$.
At the endpoint $q=\infty$, set $j_\infty(b)=b$.
For every $a,b\in M$,
\begin{equation}
    \operatorname{tr}\!\left(j_p(a)j_q(b)\right)
    =
    \left\langle
      a^*\Omega,Jb^*\Omega
    \right\rangle_{\mathcal H}.
    \label{eq:appendix-complementary-core-pairing}
\end{equation}
In particular, for fixed $b\in M$, the pairing is a normal
functional of $a$.
\end{lemma}
% \begin{lemma}[Pairing of complementary symmetric cores]
% \label{lem:complementary-core-pairing}
% Let \(1<p<\infty\), let \(q=p/(p-1)\), and let
% \(J=J_{\Omega;M}\).  For every \(a,b\in M\),
% \begin{equation}
%     \operatorname{tr}\!\left(j_p(a)j_q(b)\right)
%     =
%     \left\langle
%       a^*\Omega,Jb^*\Omega
%     \right\rangle_{\mathcal H}.
%     \label{eq:appendix-complementary-core-pairing}
% \end{equation}
% \end{lemma}

\begin{proof}
We first prove the lemma for \(p\in(1,\infty)\). Suppose that \(b\) is entire analytic for
\(\sigma^\omega\).  Trace cyclicity and
\(1/(2p)+1/(2q)=1/2\) give
\begin{align*}
    \operatorname{tr}\!\left(j_p(a)j_q(b)\right)
    &=
    \operatorname{tr}\!\left(
      h^{1/2}a h^{1/2}b
    \right)\\
    &=
    \omega\!\left(
      a\,\sigma_{-i/2}^\omega(b)
    \right).
\end{align*}
The Tomita relation
\[
    \sigma_{-i/2}^\omega(b)\Omega
    =
    Jb^*\Omega
\]
then proves \eqref{eq:appendix-complementary-core-pairing}.

For general \(b\in M\), choose a uniformly bounded net
\((b_i)\) of modular-entire elements converging to \(b\) in the
strong-\(*\) topology.  Lemma~\ref{lem:sandwich-continuity} gives
\(j_q(b_i)\to j_q(b)\) in \(L^q(M)\), while
\(Jb_i^*\Omega\to Jb^*\Omega\) in \(\mathcal H\).
Passing to the limit proves the identity.

For $p=1$, trace cyclicity and the case $p=q=2$ give
\[
    \operatorname{tr}(j_1(a)b)
    =\operatorname{tr}(j_2(a)j_2(b))
    =\langle a^*\Omega,Jb^*\Omega\rangle_{\mathcal H}.
\]
The normality assertion follows from this vector-functional
representation.
\end{proof}

\begin{lemma}[Positive \(L^p\)-elements as sandwiched R\'{e}nyi densities]
\label{lem:positive-Ln-sandwiched-density}
Let $1< p<\infty$, let $q\in(1,\infty)$ satisfy
$1/p+1/q=1$, and set $\eta_p=1/(2q)$.
Let $0\neq x\in L^p(M)_+$.  There is a
unique non-zero \(\psi_x\in M_*^+\) such that
\begin{equation}
    h_{\psi_x}
    =
    h^{\eta_p}x h^{\eta_p}.
    \label{eq:functional-from-positive-Ln}
\end{equation}
For this functional, \(x=x_{\psi_x\mid\omega,p}\), and
\begin{equation}
    Q_p(\psi_x\Vert\omega;M)
    =
    \|x\|_p^p
    =
    \operatorname{tr}(x^p).
    \label{eq:Qn-positive-Ln}
\end{equation}
\end{lemma}

\begin{proof}
Since \(h^{\eta_p}\in L^{2q}(M)\) and
\[
    \frac{1}{2q}+\frac1p+\frac{1}{2q}=1,
\]
H\"older's inequality shows that
\(h^{\eta_p}xh^{\eta_p}\) belongs to \(L^1(M)\).  It is positive because
\[
    h^{\eta_p}xh^{\eta_p}
    =
    \bigl(x^{1/2}h^{\eta_p}\bigr)^*
    \bigl(x^{1/2}h^{\eta_p}\bigr).
\]
To see that it is non-zero, suppose instead that
\(h^{\eta_p}xh^{\eta_p}=0\).  Then, for every \(a\in M\),
\[
    \operatorname{tr}\!\left(xj_{q}(a)\right)
    =
    \operatorname{tr}\!\left(
      h^{\eta_p}xh^{\eta_p}a
    \right)
    =0.
\]
Since \(j_{q}(M)\) is dense in \(L^{q}(M)\), duality forces
\(x=0\), a contradiction.  The isomorphism
\(M_*\cong L^1(M)\) therefore gives the unique non-zero functional
\(\psi_x\).  Definition~\ref{def:SRD} and uniqueness of the sandwiched
density then give \eqref{eq:Qn-positive-Ln}.
\end{proof}

\section{Construction of the Ward test core}
\label{appsec:Ward-core-construction}

This appendix constructs the test core used in
Proposition~\ref{prop:Ward-core-toolkit}.  We use the notation of
Sections~\ref{sec:Ln-dynamics} and~\ref{sec:integer-renyi-qnec}.  In particular,
\(M=M_A\), \(h=h_\omega\), \(G_p\) is the generator of
\((V_c^{(p)})_{c\geq0}\), and
\((\Sigma_s^{(p)})_{s\in\mathbb R}\) is the modular group on \(L^p(M)\) for \(1\leq p<\infty\).
We interpret \(M\)-valued integrals in the ultraweak sense and
\(L^p(M)\)-valued integrals in the Bochner sense. As before, we use \(n\geq2\) to denote an integer. 

%%%%%%%%%%%%%%%%%%%%%%%%%%%%%%%%%%%%%%%%%%%%%%%%%%%%%%%%%%%%%%%%%%%%%%%%%%%%%%
\subsection{KMS placement and compatibility with integration}
\label{appsubsec:KMS-and-integration}

\begin{proof}[Proof of \eqref{eq:KMS-placement-Ward}]
To prove the KMS placement formula, we commute powers of \(h\) past
modular-entire observables.  Expanding the symmetric embeddings and using
the cyclicity from
Lemma~\ref{lem:holder-products-trace-cyclicity} gives
\begin{align*}
    \operatorname{tr}\!\left(
      j_n(a_1)\cdots j_n(a_n)
    \right)
    &=
    \operatorname{tr}\!\left(
      a_1h^{1/n}a_2h^{1/n}\cdots a_nh^{1/n}
    \right).
\end{align*}
For a modular-entire element \(a\), the covariance relation in the
continuous core is
\begin{equation}
    h^s a
    =
    \sigma_{-is}^\omega(a)h^s,
    \qquad s\geq0;
    \label{eq:appendix-move-h-through-a}
\end{equation}
see \cite[Lemma~10.5]{Hiai:2021cks}.  Moving the successive
powers of \(h\) to the right therefore turns the operator inside the
trace into
\[
    a_1
    \sigma_{-i/n}^\omega(a_2)
    \cdots
    \sigma_{-(n-1)i/n}^\omega(a_n)\,h.
\]
Another cyclic permutation, followed by
\(\operatorname{tr}(hb)=\omega(b)\) for \(b\in M\), proves
\eqref{eq:KMS-placement-Ward}.
\end{proof}

The next lemma justifies applying \(j_p\) to the ultraweak averages used
below.

\begin{lemma}[Compatibility of \(j_p\) with integration]
\label{lem:jn-compatible-weak}
Let \((S,\nu)\) be a finite complex measure space and let
\(s\mapsto a_s\in M\) be bounded and ultraweakly measurable.  Suppose
that \(s\mapsto j_p(a_s)\) is Bochner integrable.  Then
\begin{equation}
    j_p\!\left(\int_S a_s\,d\nu(s)\right)
    =
    \int_S j_p(a_s)\,d\nu(s).
    \label{eq:jn-compatible-integration}
\end{equation}
\end{lemma}

\begin{proof}
Let \(q\) be the H\"older conjugate of \(p\), defined such that \(1/q+1/p=1\). For \(b\in M\), the map
\[
    a\longmapsto
    \operatorname{tr}\!\left(j_p(a)j_{q}(b)\right)
\]
is a normal functional on \(M\), by
Lemma~\ref{lem:complementary-core-pairing}.  Pairing the two sides of
\eqref{eq:jn-compatible-integration} with \(j_{q}(b)\) therefore gives
the same scalar integral.  Since \(j_{q}(M)\) is dense in
\(L^{q}(M)\), the duality between \(L^p(M)\) and \(L^{q}(M)\)
separates points and proves the claim.
\end{proof}

%%%%%%%%%%%%%%%%%%%%%%%%%%%%%%%%%%%%%%%%%%%%%%%%%%%%%%%%%%%%%%%%%%%%%%%%%%%%%%
\subsection{Null and modular smoothing}
\label{appsubsec:null-modular-smoothing}

Let \(C_c^\infty((0,\infty))\) denote the smooth functions compactly
supported in the open half-line.  For
\(f\in C_c^\infty((0,\infty))\), define
\begin{align}
    a_f
    &\coloneqq
    \int_0^\infty f(t)T_t(a)\,dt,
    \qquad a\in M,
    \label{eq:appendix-null-smearing}\\
    S_f^{(p)}x
    &\coloneqq
    \int_0^\infty f(t)V_t^{(p)}x\,dt,
    \qquad x\in L^p(M).
    \label{eq:appendix-Ln-null-smearing}
\end{align}

\begin{lemma}[Null smoothing]
\label{lem:null-smoothing-fixed}
For \(a\in M\), \(x\in L^p(M)\), and
\(f\in C_c^\infty((0,\infty))\), one has
\begin{align}
    a_f&\in\Dom(\delta),
    &
    \delta a_f
    &=
    -\int_0^\infty f'(t)T_t(a)\,dt,
    \label{eq:appendix-delta-af}\\
    S_f^{(p)}x&\in\Dom(G_p),
    &
    G_pS_f^{(p)}x
    &=
    -\int_0^\infty f'(t)V_t^{(p)}x\,dt.
    \label{eq:appendix-Gn-Sf}
\end{align}
Moreover,
\begin{equation}
    S_f^{(p)}j_p(a)=j_p(a_f),
    \label{eq:appendix-null-smoothing-compatibility}
\end{equation}
and
\begin{equation}
    T_u(a_f)-a_f
    =
    \int_0^u T_r(\delta a_f)\,dr,
    \qquad
    \|T_u(a_f)-a_f\|
    \leq
    u\|\delta a_f\|,
    \qquad u\geq0.
    \label{eq:appendix-null-Lipschitz}
\end{equation}
\end{lemma}

\begin{proof}
Extend \(f\) by zero to the whole real line.  For \(r>0\), the semigroup
property and a change of variables give
\begin{equation}
    \frac{T_r(a_f)-a_f}{r}
    =
    \int_0^\infty
    \frac{f(t-r)-f(t)}{r}\,T_t(a)\,dt.
    \label{eq:appendix-null-difference-quotient}
\end{equation}
The scalar difference quotient converges to \(-f'\) in \(L^1(\mathbb
R)\).  The estimate
\[
    \left\|
      \int_0^\infty g(t)T_t(a)\,dt
    \right\|
    \leq
    \|g\|_1\|a\|
\]
therefore proves \eqref{eq:appendix-delta-af}.  The same argument
in \(L^p(M)\), using contractivity and strong continuity of
\(V^{(p)}\), proves the usual semigroup smoothing formula
\eqref{eq:appendix-Gn-Sf}
\cite[proof of Proposition~II.1.8]{engel.nagel:00:one-parameter}.

Lemma~\ref{lem:jn-compatible-weak} gives
\eqref{eq:appendix-null-smoothing-compatibility}.  Finally, applying the
scalar identity
\[
    f(t-u)-f(t)
    =
    -\int_0^u f'(t-r)\,dr
\]
in \eqref{eq:appendix-null-difference-quotient}, and then using Fubini
and the semigroup property, gives the first equality in
\eqref{eq:appendix-null-Lipschitz}.  Its norm estimate follows from the
contractivity of \(T_r\).
\end{proof}

For \(\varepsilon>0\), set
\begin{equation}
    \gamma_\varepsilon(s)
    \coloneqq
    \frac{1}{\sqrt{2\pi}\,\varepsilon}
    \exp\!\left(-\frac{s^2}{2\varepsilon^2}\right).
    \label{eq:appendix-Gaussian-kernel}
\end{equation}
Define the weighted modular averages
\begin{align}
    \mathfrak M_{\varepsilon,\lambda}(a)
    &\coloneqq
    \int_{\mathbb R}
    \gamma_\varepsilon(s)e^{\lambda s}
    \sigma_s^\omega(a)\,ds,
    & a\in&\ M,\quad\lambda\in\{0,2\pi\},
    \label{eq:appendix-modular-smearing-M}\\
    \mathcal M_\varepsilon^{(p)}y
    &\coloneqq
    \int_{\mathbb R}
    \gamma_\varepsilon(s)\Sigma_s^{(p)}y\,ds,
    & y\in&\ L^p(M).
    \label{eq:appendix-modular-smearing-Ln}
\end{align}

\begin{lemma}[Gaussian modular smoothing]
\label{lem:modular-smoothing-fixed}
For \(a\in M\) and \(\lambda\in\{0,2\pi\}\), the element
\(\mathfrak M_{\varepsilon,\lambda}(a)\) is entire analytic for
\(\sigma^\omega\), with
\begin{equation}
    \sigma_z^\omega\!\left(
      \mathfrak M_{\varepsilon,\lambda}(a)
    \right)
    =
    \int_{\mathbb R}
    \gamma_\varepsilon(s-z)e^{\lambda(s-z)}
    \sigma_s^\omega(a)\,ds,
    \qquad z\in\mathbb C.
    \label{eq:appendix-entire-Gaussian-orbit}
\end{equation}
Furthermore,
\begin{equation}
    \mathcal M_\varepsilon^{(p)}j_p(a)
    =
    j_p\!\left(\mathfrak M_{\varepsilon,0}(a)\right).
    \label{eq:appendix-modular-smoothing-compatibility}
\end{equation}
If \(y\in\Dom(G_p)\), then
\(\mathcal M_\varepsilon^{(p)}y\in\Dom(G_p)\) and
\begin{equation}
    G_p\mathcal M_\varepsilon^{(p)}y
    =
    \int_{\mathbb R}
    \gamma_\varepsilon(s)e^{2\pi s}
    \Sigma_s^{(p)}G_py\,ds.
    \label{eq:appendix-Gn-modular-smoothing}
\end{equation}
As \(\varepsilon\downarrow0\),
\begin{equation}
    \mathcal M_\varepsilon^{(p)}y\longrightarrow y,
    \qquad
    G_p\mathcal M_\varepsilon^{(p)}y
    \longrightarrow G_py
    \quad\text{in }L^p(M).
    \label{eq:appendix-modular-graph-convergence}
\end{equation}
\end{lemma}

\begin{proof}
For real \(t\), changing variables in
\(\sigma_t^\omega(\mathfrak M_{\varepsilon,\lambda}(a))\) gives the right-hand side of
\eqref{eq:appendix-entire-Gaussian-orbit} at \(z=t\).  The map
\[
    z\longmapsto
    \gamma_\varepsilon(\,\cdot-z\,)
    e^{\lambda(\,\cdot-z\,)}
\]
is entire as an \(L^1(\mathbb R)\)-valued function because its derivatives
have locally uniform integrable majorants.  Moreover,
\[
    \left\|
      \int_{\mathbb R}g(s)\sigma_s^\omega(a)\,ds
    \right\|
    \leq
    \|g\|_1\|a\|.
\]
Thus the integral on the right-hand side of
\eqref{eq:appendix-entire-Gaussian-orbit} is norm-entire in \(z\) and
extends the real modular orbit.  Equation
\eqref{eq:appendix-modular-smoothing-compatibility} follows from
Lemma~\ref{lem:jn-compatible-weak}.

For \(y\in\Dom(G_p)\), the covariance
\eqref{eq:Ln-generator-modular-covariance} gives
\[
    G_p\Sigma_s^{(p)}y
    =
    e^{2\pi s}\Sigma_s^{(p)}G_py.
\]
After multiplication by \(\gamma_\varepsilon(s)\), this expression is
Bochner integrable, as is \(\gamma_\varepsilon(s)\Sigma_s^{(p)}y\).  Since
\(G_p\) is closed, it may be passed through the integral
\cite[Proposition~C.4]{engel.nagel:00:one-parameter}, which proves
\eqref{eq:appendix-Gn-modular-smoothing}.

The first convergence in
\eqref{eq:appendix-modular-graph-convergence} follows from strong
continuity of \(\Sigma^{(p)}\) and the fact that
\(\gamma_\varepsilon(s)\,ds\) is an approximate identity.  For the
second, complete the square:
\begin{equation}
    \gamma_\varepsilon(s)e^{2\pi s}
    =
    e^{2\pi^2\varepsilon^2}
    \gamma_\varepsilon(s-2\pi\varepsilon^2).
    \label{eq:appendix-shifted-Gaussian}
\end{equation}
The measures on the right have mass tending to one and concentrate at
the origin.  Strong continuity of
\(s\mapsto\Sigma_s^{(p)}G_py\) in
\eqref{eq:appendix-Gn-modular-smoothing} therefore proves the second
convergence.
\end{proof}

%%%%%%%%%%%%%%%%%%%%%%%%%%%%%%%%%%%%%%%%%%%%%%%%%%%%%%%%%%%%%%%%%%%%%%%%%%%%%%
\subsection{Construction and graph density}
\label{appsubsec:Ward-test-space}

For \(a\in M\), \(f\in C_c^\infty((0,\infty))\), and
\(\varepsilon>0\), put
\begin{equation}
    a_{f,\varepsilon}
    \coloneqq
    \mathfrak M_{\varepsilon,0}(a_f).
    \label{eq:appendix-double-smearing}
\end{equation}

\begin{definition}[Ward test space]
\label{def:smooth-core-null}
Let
\begin{equation}
    \mathcal W
    \coloneqq
    \operatorname{span}
    \left\{
      a_{f,\varepsilon}:
      a\in M,\ 
      f\in C_c^\infty((0,\infty)),\
      \varepsilon>0
    \right\},
    \qquad
    \mathcal D_p^{\mathrm W}\coloneqq j_p(\mathcal W).
    \label{eq:appendix-Ward-test-space}
\end{equation}
\end{definition}

\begin{proof}[Proof of Proposition~\ref{prop:Ward-core-toolkit}]
We first establish the analytic and covariance properties.  For 
\(a_{f,\varepsilon}\in\mathcal W\), real Borchers
covariance \eqref{eq:Tc-modular-covariance} gives
\begin{align}
    &\frac{T_r(a_{f,\varepsilon})-a_{f,\varepsilon}}{r}
    =
    \int_{\mathbb R}
      \gamma_\varepsilon(s)e^{2\pi s}
      \sigma_s^\omega\!\left(
        \frac{T_{re^{2\pi s}}(a_f)-a_f}
             {re^{2\pi s}}
      \right)ds.
    \label{eq:appendix-double-smear-difference-quotient}
\end{align}
The integrand converges pointwise in operator norm to
\[
    \gamma_\varepsilon(s)e^{2\pi s}
    \sigma_s^\omega(\delta a_f).
\]
By \eqref{eq:appendix-null-Lipschitz}, the norm of the expression in
parentheses in
\eqref{eq:appendix-double-smear-difference-quotient} is at most
\(\|\delta a_f\|\).  Dominated convergence therefore proves that
\(a_{f,\varepsilon}\in\Dom(\delta)\) and
\begin{equation}
    \delta a_{f,\varepsilon}
    =
    \int_{\mathbb R}
      \gamma_\varepsilon(s)e^{2\pi s}
      \sigma_s^\omega(\delta a_f)\,ds
    =
    \mathfrak M_{\varepsilon,2\pi}(\delta a_f).
    \label{eq:appendix-delta-double-smear}
\end{equation}
Lemma~\ref{lem:modular-smoothing-fixed} now shows that both
\(a_{f,\varepsilon}\) and \(\delta a_{f,\varepsilon}\) are modular
entire.  In particular,
\begin{equation}
    \sigma_z^\omega(\delta a_{f,\varepsilon})
    =
    \int_{\mathbb R}
      \gamma_\varepsilon(s-z)e^{2\pi(s-z)}
      \sigma_s^\omega(\delta a_f)\,ds.
    \label{eq:appendix-entire-delta-double-smear}
\end{equation}

It remains to check complex covariance.  Combining
\eqref{eq:appendix-entire-Gaussian-orbit} with real Borchers covariance
and repeating the dominated-convergence argument above gives
\begin{align}
    \delta\!\left(
      \sigma_z^\omega(a_{f,\varepsilon})
    \right)
    &=
    \int_{\mathbb R}
      \gamma_\varepsilon(s-z)e^{2\pi s}
      \sigma_s^\omega(\delta a_f)\,ds
    \notag\\
    &=
    e^{2\pi z}
    \sigma_z^\omega(\delta a_{f,\varepsilon}).
    \label{eq:appendix-complex-Borchers}
\end{align}
Indeed, for fixed \(z\), the differences between the relevant
difference quotients and their pointwise limits are dominated by
\[
    2|\gamma_\varepsilon(s-z)|e^{2\pi s}
    \|\delta a_f\|,
\]
which is integrable.  This also proves
\(\sigma_z^\omega(a_{f,\varepsilon})\in\Dom(\delta)\).
Linearity establishes the same conclusions for every \(a\in\mathcal W\).

If \(a\in\mathcal W\), operator-norm convergence in the definition of
\(\delta a\), together with H\"older's inequality, gives
\[
    \frac{V_r^{(p)}j_p(a)-j_p(a)}{r}
    =
    j_p\!\left(
      \frac{T_r(a)-a}{r}
    \right)
    \longrightarrow
    j_p(\delta a)
    \quad\text{in }L^p(M).
\]
Thus \(\mathcal D_p^{\mathrm W}\subset\Dom(G_p)\) and
\begin{equation}
    G_pj_p(a)=j_p(\delta a),
    \qquad a\in\mathcal W.
    \label{eq:appendix-Gn-on-Ward-space}
\end{equation}

We finish by proving graph density.  Choose
\(\rho\in C_c^\infty((0,1))\) with \(\rho\geq0\) and
\(\int_0^\infty\rho(t)\,dt=1\), and set
\[
    \rho_m(t)\coloneqq m\rho(mt).
\]
For \(y\in\Dom(G_p)\), strong continuity, orbit regularity
\eqref{eq:semigroup-orbit-regularity}, and the closedness of \(G_p\)
under Bochner integration imply
\begin{equation}
    S_{\rho_m}^{(p)}y\longrightarrow y,
    \qquad
    G_pS_{\rho_m}^{(p)}y
    =
    S_{\rho_m}^{(p)}G_py
    \longrightarrow G_py.
    \label{eq:appendix-semigroup-graph-approximation}
\end{equation}
For each \(m\), use the density of \(j_p(M)\) to choose
\(z_m=j_p(a_m)\) so that
\begin{equation}
    \bigl(1+\|\rho_m'\|_1\bigr)
    \|z_m-y\|_p<\frac1m.
    \label{eq:appendix-fast-core-approximation}
\end{equation}
Contractivity and \eqref{eq:appendix-Gn-Sf} yield
\begin{align*}
    \left\|
      S_{\rho_m}^{(p)}(z_m-y)
    \right\|_p
    &\leq
    \|z_m-y\|_p,%\\
    \qquad
    \left\|
      G_pS_{\rho_m}^{(p)}(z_m-y)
    \right\|_p
    \leq
    \|\rho_m'\|_1\|z_m-y\|_p.
\end{align*}
Together with
\eqref{eq:appendix-semigroup-graph-approximation}, this shows that
\[
    w_m
    \coloneqq
    S_{\rho_m}^{(p)}z_m
    =
    j_p\!\left((a_m)_{\rho_m}\right)
\]
converges to \(y\) in the graph norm of \(G_p\).

For each fixed \(m\),
Lemma~\ref{lem:modular-smoothing-fixed} gives
\[
    \mathcal M_\varepsilon^{(p)}w_m
    \longrightarrow w_m
    \quad\text{in graph norm as }\varepsilon\downarrow0.
\]
Choose \(\varepsilon_m>0\) so that this graph-norm difference is less
than \(1/m\).  By
\eqref{eq:appendix-modular-smoothing-compatibility},
\[
    \mathcal M_{\varepsilon_m}^{(p)}w_m
    =
    j_p\!\left(
      (a_m)_{\rho_m,\varepsilon_m}
    \right)
    \in\mathcal D_p^{\mathrm W}.
\]
These elements converge to \(y\) in graph norm.  Hence
\(\mathcal D_p^{\mathrm W}\) is a graph core for \(G_p\).
\end{proof}

\subsection{Results for the integer R\'enyi ant formula}

The next two results based on the Ward test core are used in Section \ref{sec:ant-formula} to prove the integer R\'enyi ant formula.

\begin{lemma}[Power rule on the Ward test core]
\label{lem:ward-power-core}
Let $n\geq2$ be an integer, let $1\leq j\leq n$ be an
integer, and set $r=n/j$. For $a\in\mathcal W$, put
\[
    x=j_n(a),\qquad u=G_nx=j_n(\delta a).
\]
Then $x^j\in\operatorname{Dom}(G_r)$ and
\begin{equation}
    G_r(x^j)
    =
    \sum_{\ell=0}^{j-1}
    e^{\pi i(j-1-2\ell)/n}
    x^\ell u x^{j-1-\ell}.
    \label{eq:ward-power-core}
\end{equation}
In particular, if $n=2m$, then
$x^m\in\operatorname{Dom}(G_2)$ and
\begin{equation}
    G_2(x^m)
    =
    \sum_{\ell=0}^{m-1}
    e^{\pi i(m-1-2\ell)/n}
    x^\ell u x^{m-1-\ell}.
    \label{eq:ward-half-power-core}
\end{equation}
\end{lemma}

\begin{proof}
For $0\leq\ell\leq j-1$, set
\[
    \vartheta_\ell
    :=\frac{i(j-1-2\ell)}{2n},
    \qquad
    b:=\prod_{\ell=0}^{j-1}
       \sigma^\omega_{\vartheta_\ell}(a),
\]
where the factors are ordered from left to right by
increasing $\ell$. Commuting powers of $h$ past
modular-entire observables gives
\[
    x^j=j_r(b).
\]
More generally, for modular-entire $a_0,\ldots,a_{j-1}$,
\begin{equation}\label{eq:mixed-prod-identity}
    \prod_{\ell=0}^{j-1}j_n(a_\ell)
    =
    j_r\!\left(
        \prod_{\ell=0}^{j-1}
        \sigma^\omega_{\vartheta_\ell}(a_\ell)
    \right).
\end{equation}

The properties of $\mathcal W$ in Proposition~\ref{prop:Ward-core-toolkit} imply
\[
    \sigma^\omega_{\vartheta_\ell}(a)
    \in\operatorname{Dom}(\delta),
    \qquad
    \delta\bigl(\sigma^\omega_{\vartheta_\ell}(a)\bigr)
    =
    e^{2\pi\vartheta_\ell}
    \sigma^\omega_{\vartheta_\ell}(\delta a).
\]
Since $\operatorname{Dom}(\delta)$ is an algebra,
$b\in\operatorname{Dom}(\delta)$. Boundedness of
$j_r:M\to L^r(M)$ therefore gives
\[
    \frac{V_c^{(r)}j_r(b)-j_r(b)}{c}
    =
    j_r\!\left(\frac{T_c(b)-b}{c}\right)
    \longrightarrow j_r(\delta b)
    \quad\text{in }L^r(M)
\]
as $c\downarrow0$. Thus $x^j\in\operatorname{Dom}(G_r)$
and $G_r(x^j)=j_r(\delta b)$.

Applying the Leibniz rule to $b$ gives
\[
\begin{aligned}
    \delta b
    =\sum_{\ell=0}^{j-1}
    e^{2\pi\vartheta_\ell}
    \sigma^\omega_{\vartheta_0}(a)
    \cdots
    \sigma^\omega_{\vartheta_{\ell-1}}(a)
    \sigma^\omega_{\vartheta_\ell}(\delta a)
    \sigma^\omega_{\vartheta_{\ell+1}}(a)
    \cdots
    \sigma^\omega_{\vartheta_{j-1}}(a),
\end{aligned}
\]
with empty products omitted. Applying $j_r$ and using
the mixed-product identity~\eqref{eq:mixed-prod-identity} yields
\eqref{eq:ward-power-core}, because
\[
    e^{2\pi\vartheta_\ell}
    =e^{\pi i(j-1-2\ell)/n}.
\]
Taking $j=m$ when $n=2m$ gives
\eqref{eq:ward-half-power-core}.
\end{proof}

\begin{lemma}[Half-power trace identity]
\label{lem:ward-half-power-domain}
Let $n\geq2$ be an integer and let
$x\in\operatorname{Dom}(G_n)$. Put $u=G_nx$.
If $n$ is odd, assume additionally that $x\geq0$
and $x^{1/2}\in\operatorname{Dom}(G_{2n})$.
Then $x^{n/2}\in\operatorname{Dom}(G_2)$ and
\begin{equation}
    \operatorname{tr}\bigl(x^{n/2}G_2(x^{n/2})\bigr)
    =
    \frac{1}{\sin(\pi/n)}
    \operatorname{tr}\bigl(u x^{n-1}\bigr).
    \label{eq:ward-half-power-trace}
\end{equation}
For $n=2m$, \eqref{eq:ward-half-power-core} also holds.
\end{lemma}

\begin{proof}
We first extend \eqref{eq:ward-power-core} to
generator domains. For any integer input exponent
$n\geq2$, graph density gives
$x_k\in j_n(\mathcal W)$ such that
\[
    x_k\longrightarrow x,\qquad
    G_nx_k\longrightarrow G_nx
    \quad\text{in }L^n(M).
\]
For each integer $1\leq j\leq n$, Hölder's inequality
gives convergence of $x_k^j$ and of every term in
\eqref{eq:ward-power-core} in $L^{n/j}(M)$.
Closedness of $G_{n/j}$ therefore extends that
identity to every $x\in\operatorname{Dom}(G_n)$.

If $n=2m$, apply the extended identity with $j=m$.
This gives $x^m\in\operatorname{Dom}(G_2)$ and
\eqref{eq:ward-half-power-core}. By cyclicity,
\[
\begin{aligned}
    \operatorname{tr}\bigl(x^mG_2(x^m)\bigr)
    &=
    \left(
        \sum_{\ell=0}^{m-1}
        e^{\pi i(m-1-2\ell)/n}
    \right)
    \operatorname{tr}\bigl(u x^{n-1}\bigr)\\
    &=
    \frac{1}{\sin(\pi/n)}
    \operatorname{tr}\bigl(u x^{n-1}\bigr).
\end{aligned}
\]

If $n$ is odd, put $y=x^{1/2}$, $v=G_{2n}y$,
and $\theta=\pi/(2n)$. Apply the extended power rule~\eqref{eq:ward-power-core} at input exponent $2n$ to $G_n(y^2)$ gives
\begin{equation}\label{eq:power-rule-1}
    u=e^{i\theta}vy+e^{-i\theta}yv.
\end{equation}
Similarly, applying the rule at input exponent $2n$ to $G_2(y^n)$ gives
\begin{equation}\label{eq:power-rule-2}
    x^{n/2}=y^n\in\operatorname{Dom}(G_2),\quad G_2(y^n)=\sum_{\ell=0}^{n-1} e^{\pi i(n-1-2\ell)/(2n)} y^\ell v y^{n-1-\ell}. 
\end{equation}
Writing $B=\operatorname{tr}(v y^{2n-1})$, cyclicity, \eqref{eq:power-rule-1}, \eqref{eq:power-rule-2},
and the finite geometric sum give
\[
    \operatorname{tr}\bigl(u x^{n-1}\bigr)
    =2\cos\theta\,B,
    \qquad
    \operatorname{tr}\bigl(y^nG_2(y^n)\bigr)
    =\frac{B}{\sin\theta}.
\]
The identity
$2\sin\theta\cos\theta=\sin(\pi/n)$ proves
\eqref{eq:ward-half-power-trace}.
\end{proof}

\end{appendices}

{\small
\bibliographystyle{JHEP}

\providecommand{\href}[2]{#2}\begingroup\raggedright\endgroup

}

\end{document}